\documentclass[letterpaper,onecolumn,11pt,accepted=2023-10-03]{quantumarticle}
\pdfoutput=1

\usepackage[]{graphicx}
\usepackage{amssymb}
\usepackage{enumitem}
\usepackage{hyperref}
\usepackage{amsmath}
\usepackage{amsthm}
\usepackage{mathtools}
\usepackage{multicol}
\usepackage{qcircuit}
\usepackage{braket}
\usepackage{xcolor}
\usepackage[capitalize]{cleveref}
\usepackage{physics}
\usepackage{comment}
\usepackage{url}

\newcommand{\eps}{\varepsilon}

\newtheorem{theorem}{Theorem}
\newtheorem{corollary}[theorem]{Corollary}
\newtheorem{definition}[theorem]{Definition}
\newtheorem{lemma}[theorem]{Lemma}
\newtheorem{remark}[theorem]{Remark}
\newtheorem*{remark*}{Remark}

\newtheorem{proposition}[theorem]{Proposition}

\newtheorem{claim}[theorem]{Claim}

\title{Thermal State Preparation via Rounding Promises}

\author{Patrick Rall}
\affiliation{IBM Quantum, MIT-IBM Watson AI Lab, Cambridge, Massachusetts 02142, USA}
\email{patrickjrall@ibm.com}
\author{Chunhao Wang}
\affiliation{Department of Computer Science and Engineering, Pennsylvania State University}
\email{cwang@psu.edu}
\author{Pawel Wocjan}
\affiliation{IBM Quantum, Thomas J Watson Research Center, Yorktown Heights, New York 10598, USA}
\email{pawel.wocjan@ibm.com}

\begin{document}

\maketitle
\begin{abstract} 
    A promising avenue for the preparation of Gibbs states on a quantum computer is to simulate the physical thermalization process. The Davies generator describes the dynamics of an open quantum system that is in contact with a heat bath. Crucially, it does not require simulation of the heat bath itself, only the system we hope to thermalize. Using the state-of-the-art techniques for quantum simulation of the Lindblad equation, we devise a technique for the preparation of Gibbs states via thermalization as specified by the Davies generator.
    
    In doing so, we encounter a severe technical challenge: implementation of the Davies generator demands the ability to estimate the energy of the system unambiguously. That is, each energy of the system must be deterministically mapped to a unique estimate. Previous work showed that this is only possible if the system satisfies an unphysical `rounding promise' assumption. We solve this problem by engineering a random ensemble of rounding promises that simultaneously solves three problems: First, each rounding promise admits preparation of a `promised' thermal state via a Davies generator. Second, these Davies generators have a similar mixing time as the ideal Davies generator.  Third, the average of these promised thermal states approximates the ideal thermal state.

\end{abstract}

\newpage
\tableofcontents
\newpage

\section{Introduction}

\paragraph{Motivation.}

Preparing Gibbs states is a major task for quantum computers. There are several reasons for this.
First, the Gibbs state is one of the most important states of matter. For quantum models comprised of many locally interacting particles, it describes a wide range of physical situations, relevant to condensed matter physics, high energy physics, quantum chemistry \cite{2204.08349}.  Therefore, to use the quantum computer as a universal simulator of quantum systems, it is desirable to be able to prepare Gibbs states. 
Second, for general Hamiltonians, the Gibbs state is a crucial ingredient in some quantum algorithms such as those for solving semidefinite programs~\cite{AGGdW17, BKLLSW17, AG18, 1806.01838} and for training quantum Boltzman machines~\cite{KW17}.  
Third, the problem of estimating the quantum partition function, which is connected to the problem of approximately preparing the Gibbs state, plays an important role in quantum complexity theory \cite{2110.15466}.

\paragraph{Background and overview of prior work.}
There are three main approaches to preparing Gibbs states.  

The first one is a Grover-based approach in which an initial state is mapped onto a certain purification of the Gibbs state at inverse temperature $\beta$ (see \cite{0905.2199,CS17}).  The resulting running time is dictated by the overlap between the initial and target states, which is exponentially small in the system size. The Gibbs state can be approximately prepared in time on the order of $\sqrt{D/\mathcal{Z}_\beta}$, where $D$ denotes the dimension of the quantum system and $\mathcal{Z}_\beta$ the partition function at inverse temperature $\beta$.  In \cite{Holmes22} a quantum algorithm is presented for preparing a purification of the Gibbs state for the Hamiltonian $H_1=H_0+V$ at inverse temperature $\beta$ starting from a purification of the thermal state of $H_0$. 

The second approach is based on Davies generators, which is a differential equation that describes how nature thermalizes a quantum system to its thermal equilibrium \cite{Davies76,Davies79}.  Davies generators are special cases of Lindbladians \cite{Lindblad76,BP02}, which describe the most general continuous-time Markovian dynamics of an open quantum system, i.e., a quantum system that is weakly coupled to the environment and the dynamics of the environment are fast enough so that the information only flows from the system to the environment while no information is flowing back to the system. The method in \cite{2112.07646} simulates a Davies generator by attaching a heat bath and simulating time evolution on the joint system while repeatedly refreshing the bath. Their algorithm in some sense follows the original derivation of the Davies generator as the limit of such joint system-bath Hamiltonian time evolution.  When the system Hamiltonian satisfies the Eigenstate Thermalization Hypothesis (ETH), they show that the implemented quantum map not only converges to the desired Gibbs state, but also does so in polynomial time.

The third approach is based on the quantum Metropolis algorithm \cite{TVOVP11}.  This approach also avoids the exponential scaling when the system Hamiltonian satisfies ETH \cite{SC22,2112.07646}.

The advantages of the second approach are two fold. First, for every quantum system that thermalizes fast in nature, it is expected that our algorithms can also prepare its corresponding Gibbs state efficiently without suffering from the exponential dependence on the number of qubits. Second, this approach fits well for many physics-motivated applications. For example, we can use our algorithm to prepare a ``partially thermalized'' (in the natural thermalization process) thermal state, which might be of interest in some scenarios.

\bigskip

\paragraph{Main result.} 
The present work examines how to approximately prepare Gibbs states for arbitrary system Hamiltonians by simulating time evolutions according to carefully engineered Lindbladians. These Lindbladians are derived from an ideal Davies generator having the Gibbs state as unique fixed point.

There are some similarities but also important differences to the work \cite{2112.07646}.  Obviously, both are based on Davies generators.  In contrast, we do not approximate Lindblad evolution with the help of the Hamiltonian evolution of the system and a bath, a method with which it is \emph{provably impossible} to achieve linear scaling in evolution time (see~\cite{1612.09512}).  Instead, we rely on a method for directly simulating Lindbladian time evolution specified by any jump operators.  Using this method can lead to a reduction in resources: specifically, we achieve linear scaling in the mixing time $t_\mathrm{mix}$. In addition, a direct Lindblad simulation approach avoids the complication of dealing with the dynamics of the bath and the interaction Hamiltonians. We also seek to prove that our quantum map approximates the Gibbs state for any system Hamiltonian, that is, we do not need to make an assumption such as ETH. 

The quantum algorithms for simulating Lindbladian time evolution in \cite{1612.09512,LW22} assume that the jump operators have been suitably encoded. Unfortunately, it is not possible to construct jump operators of a Davies generator due to inherent imperfections of energy estimation of general Hamiltonians.  However, we show how to construct a family of Lindbladians from the given Davies generators such that simulating them with the help of the simulation algorithms in \cite{1612.09512,LW22} and taking the average of the resulting quantum states provides a good approximation of the Gibbs state. This is formulated in more detail in the theorem statement below.

\begin{theorem}[Main result -- informal statement] \label{thm:mainresult}
Assume we are given block encodings of the Hamiltonian $H$, coupling operators $S_\alpha$, and a filter function $G$. With appropriately chosen $S_\alpha$, these give rise to a Davies generator $\mathcal{L}$ that has the Gibbs state $\rho_\beta$ for inverse temperature $\beta$ as a unique fixed point.  Assume that after time $t_{\mathrm{mix}}$ the time evolved state $\exp(t_{\mathrm{mix}} \mathcal{L})(\sigma_0)$ is $\varepsilon$-close to the Gibbs state $\rho_\beta$ for any initial state $\sigma_0$.

We engineer a certain family of $2^r$ many Lindbladians $\tilde{\mathcal{L}}^{(j)}$ from the above Davies generator $\mathcal{L}$.  These Lindbladians have a new `attenuated' mixing time $t_\text{mix}$, and their jump operators can be encoded efficiently with imperfect energy estimation.  This makes it possible to simulate their time evolutions $\exp(t \tilde{\mathcal{L}}^{(j)} )$. We prove that the average
\begin{align}
    \frac{1}{2^r} \sum_j \exp(\tilde{t}_{\mathrm{mix}} \tilde{\mathcal{L}}^{(j)}) (\sigma_0)
\end{align}
is $\varepsilon$-close to the Gibbs state $\rho_\beta$ for any initial state $\sigma_0$.  Furthermore, we show how to implement the time evolution according to $\tilde{\mathcal{L}}^{(j)}$ to arbitrarily small failure probability $\delta_\mathcal{L}$, and that the total number of invocations of the block encoding of the Hamiltonian is bounded by:%\footnote{Recent work \cite{2303.18224} resolves an open question raised in this manuscript (see \cref{sec:open_questions}) on utilizing the block encoding of the form $\sum_j\ket{j}\otimes L_j$. Using the Theorem III.1 of \cite{2303.18224} together with slight adaptation of the simulation algorithm in~\cite{LW22}, the complexity of our algorithm can be improved to $\tilde O(\tilde t_\mathrm{mix} \cdot \gamma^{-1} \cdot \beta \varepsilon^{-2} )$. We also note that the additional factor of $\log^2(\beta/\eps)$ can be removed via existing techniques from \cite{2103.09717}.}
%Recent work \cite{2303.18224} resolves an open question raised in this manuscript, implying that this can be improved to $\tilde O(\tilde t_\mathrm{mix} \cdot \gamma^{-1} \cdot \beta \varepsilon^{-2} )$. We also note that the additional factor of $\log^2(\beta/\eps)$ can be removed via existing techniques from \cite{2103.09717}.}
\begin{align}
         O\left(\tilde t_\mathrm{mix} \cdot \gamma^{-1} \cdot \beta^3 \varepsilon^{-7}  \cdot \mathrm{polylog}(  \tilde t_\mathrm{mix}/\delta_\mathcal{L}) \cdot \log^2(\beta/\varepsilon) \right),
\end{align}
where $\gamma$ is an attenuation coefficient that affects the attenuated mixing time $\tilde t_\mathrm{mix}$.
\end{theorem}

\begin{remark*}
    After the first version of this manuscript was made public, recent work \cite{2303.18224} resolved an open question raised in this manuscript (see \cref{sec:open_questions}) on utilizing the block encoding of the form $\sum_j\ket{j}\otimes L_j$. Using the Theorem III.1 of \cite{2303.18224} together with slight adaptation of the simulation algorithm in~\cite{LW22}, the complexity of our algorithm can be improved to $\tilde O(\tilde t_\mathrm{mix} \cdot \gamma^{-1} \cdot \beta \varepsilon^{-2} )$. We also note that the additional factor of $\log^2(\beta/\eps)$ can be removed via existing techniques from \cite{2103.09717}.
\end{remark*}

In Section~\ref{sec:lindbladdynamicsonthepromisedsubspace} we give some numerical experiments indicating that the attenuated mixing time $\tilde{t}_{\mathrm{mix}}$ is on the order of the original mixing time $t_\mathrm{mix}$ for suitably chosen $\gamma^{-1}$. 

We formulate in \cref{sec:open_questions} some open research questions whose solution could lead to improvements of our current methods. 

\paragraph{Technical overview.}
As mentioned above, the implementation of jump operators of a Davies generator requires perfect energy estimation. More precisely, what we mean by perfect is that energy estimation would have to be unambiguous: for each energy of the Hamiltonian, it must yield a unique energy estimate.  Unfortunately, energy estimation unavoidably produces superpositions of different energy estimates for general Hamiltonians: if $\ket{\psi}$ is an eigenstate of the Hamiltonian with eigenvalue $\lambda$, energy estimation implements a map:
\begin{align}
    \ket{\psi} \hspace{2mm}\mapsto\hspace{2mm} \ket{\psi} \otimes \left(\alpha \ket*{\tilde \lambda_1} + \beta \ket*{\tilde \lambda_2} \right)
\end{align}
where $\tilde \lambda_1$ and $\tilde \lambda_2$ are two different estimates of $\lambda$.   This superposition of estimates cause constant-size errors in quantum algorithms implementing Davies generators, and must be eliminated. With perfect estimation, the superposition on the second register is not present and there is always a unique estimate. 

It is possible to construct an `approximate Davies generator' from an energy estimation algorithm that produces superpositions of estimates. However, the resulting dynamics no longer correspond to the Davies generator of any particular Hamiltonian, making it challenging to analyze. To our knowledge, no method exists for rigorously proving the accuracy of an algorithm based on such an approximate Davies generator. We highlight some of the challenges of this task in \cref{app:approximateDaviesgenerator}. 

It was shown in \cite{2103.09717} that perfect energy estimation is possible when the Hamiltonian satisfies a `rounding promise' assumption. The rounding promise prohibits the Hamiltonian from having any eigenvalues that induce a superposition of estimates. However, such an assumption on the Hamiltonian is extremely unphysical and will not be satisfiable in practice. The main technical idea of the present work is to shift the notion of a rounding promise away from the Hamiltonian itself, but rather to a family of states. The basic idea is very simple: if a quantum state has no support on any eigenstates whose eigenvalues have a superposition of estimates, then it is as if the Hamiltonian did not have such eigenvalues. This family of states is defined by a `promised subspace'.

Since perfect energy estimation is possible on a promised subspace, implementation of Davies generators is possible as well. While the analysis involves a wide variety of error parameters, we find that all of them admit a mathematically rigorous treatment. 

Our construction relies on just one assumption: that we can construct coupling operators for each promised subspace that ensure that the Davies generator converges in a reasonable amount of time. We give numerical evidence that projecting a coupling operator into a promised subspace does not significantly reduce the Davies generator's convergence time. Notice how this assumption has nothing to do with the protocol's accuracy, only its convergence time.

We now summarize the different types of errors that occur in our implementation, and the rough ideas behind keeping them under control:

\begin{itemize}

    \item \textbf{Lindblad simulation error}:
      There exists a quantum algorithm, e.g.~\cite{1612.09512}, for simulating Lindblad evolution given a description of the jump operators of a Lindbladian $\mathcal{L}$. For evolution time $t$ and accuracy $\delta_{\mathrm{sim}}$, this algorithm implements a quantum channel with accuracy $\delta_{\mathrm{sim}}$ in terms of the diamond norm to the ideal channel $e^{\mathcal{L}t}$. In our case, we simulate the dynamics that approximates a special Lindbladian called the Davies generator.
    
    \item \textbf{Knowledge of mixing time}:
    Given a Hamiltonian $H$, coupling operators $S_\alpha$, and a filter function $G$, the Davies generator is a certain Lindbladian whose fixed point is the desired thermal state $\rho_\beta=e^{-\beta H}/\Tr e^{-\beta H}$. We assume 
    we are given the Hamiltonian and suitable coupling operators of the Davies generator together with a bound $t_{\mathrm{mix}}$ after which, the state is sufficiently close the desired thermal state.
    
    \item \textbf{Precision and failure probability of energy estimation}: The jump operators of a Davies generator depend on the coupling operators and on the Bohr frequencies (differences of the energies) of the Hamiltonian and the corresponding pairs of eigensubspaces.  To realize these, we rely on energy estimation of the Hamiltonian.  
    
    By restricting to the promised subspace, we can guarantee using techniques from \cite{2103.09717} that for every energy $\lambda$ there exists a single unique estimate $\tilde \lambda$. There remain two other kinds of error, both of which can be dealt with rigorously. 
    
    First, the energy estimation map is only implemented with a failure probability  $\delta_{\mathrm{est}}$ that can be exponentially suppressed. By phrasing this error as a deviation in spectral norm of the Lindbladian jump operators, we can show that this error is not blown up when the Davies generator is simulated for a long period of time.
    
    Second, the cost of energy estimation scales linearly with the precision of the estimate. To analyze this, we leverage a trick from \cite{0905.2199}: instead of preparing the thermal state for the original Hamiltonian, one can interpret the resulting process as preparing the thermal state for a rounded Hamiltonian.  Since the norm of the difference of these Hamiltonians is small, the corresponding thermal states are close to each other.

    \item \textbf{Preparation of an initial state on the promised subspace}: We implement a Davies generator whose dynamics are trivial outside the promised subspace. In order to prepare a promised thermal state, we must feed the Davies generator with a state that is exclusively supported on the promised subspace. We achieve this by taking an arbitrary input state and measuring a two-outcome `left-right' POVM. Depending on the POVM outcome, we know that either a `left' or a `right' rounding promise is satisfied by the post-measurement state.
    
    \item \textbf{Approximation of the ideal Gibbs state}: Our goal is to prepare a Gibbs state supported on the entire Hilbert space. But our protocol only prepares promised Gibbs states, which are only supported on the promised subspace. Thus, individual promised thermal states are generically far in trace distance from the ideal state. To deal with this, we show that there exists an ensemble of rounding promises such that the ensemble average of all the promised Gibbs states is an accurate approximation of the ideal thermal state. The basic idea of this ensemble is that the probability of any particular energy being excluded can be made arbitrarily small.
    
    \item \textbf{Leakage and attenuation errors in removal of rounding promise}:
    The constructed coupling operators of the Davies generators for the well-rounded Hamiltonians on the promised subspaces have two important types of errors, namely, `leakage' and `attenuation.' 
    
    The leakage error $\delta_{\mathrm{leak}}$ measures how much coupling operators and initial states `leak' outside a promised space. Fortunately, $\delta_{\mathrm{leak}}$ can be made exponentially small by using quantum singular value transformation.  
    
    Blocks of coupling operators corresponding to some pairs of energies can be `attenuated', i.e., multiplied by small positive numbers.  However, as long as attenuation remains non-zero, the fixed-point remains the thermal state for the well-rounded Hamiltonian on the promised subspace.  Unfortunately, the mixing speed can be negatively affected. 
    
    \item \textbf{Mixing assumption for projected/attenuated coupling operators:} In principle, it could happen that the ideally projected coupling operators do not guarantee convergence to the thermal state of the well-rounded Hamiltonian on the promised subspace anymore.  Moreover, even if they do, attenuation could increase the convergence time.
    
    However, we provide numerical evidence that these unfavorable situations do not typically occur.  For the theoretical analysis of our quantum algorithm, we must assume that the attenuated coupling operators for the well-rounded Hamiltonian on the promised subspaces have similar mixing behavior as the original coupling operators.

\end{itemize}

\section{Preliminaries}

In this section, we first review some preliminaries about Gibbs states and Davies generators in order to establish notation and to rigorously define our goal: to prepare a Gibbs state by simulating the time evolution of a Davies generator. To do so, we leverage an algorithm for simulating general Lindblad evolution \cite{LW22}.

In order to implement the Davies generator, we require conditions under which the energy of a Hamiltonian can be estimated without producing superpositions of different energy estimates. %rounding errors. 
So, in \cref{sec:roundingpromises} we establish the notion of a rounding promise, and review a result from \cite{2103.09717} how a rounding promise can ensure that each energy is rounded to a unique energy when performing energy estimation.

Here are some notations and conventions used in this paper. We use $\mathbb{Z}_+$ to denote the set of all positive integers. Let $\mathcal{H}$ be a Hilbert space. We use $\mathrm{L}(\mathcal{H})$ to denote the collection of all linear operators (matrices) of the form: $A: \mathcal{H}\rightarrow \mathcal{H}$. For a matrix $A$, the \emph{spectral norm} $\norm{A}$ is the largest singular value of $A$, and the \emph{trace norm} $\norm{A}_1$ is the sum of the singular values. For a superoperator $\Lambda$, the \emph{induced trace norm} of $\Lambda$, denoted by $\norm{\Lambda}_1$, is defined as $\norm{\Lambda}_1 = \max_{A \neq 0} \norm{\Lambda(A)}_1/\norm{A}_1$. If $\Lambda$ is acting on $\mathrm{L}(\mathcal{H})$ for some Hilbert space $\mathcal{H}$, then, the \emph{diamond-norm} of $\Lambda$, denoted by $\norm{\Lambda}_{\diamond}$, is defined as $\norm{\Lambda}_{\diamond} = \norm{\Lambda\otimes \mathcal{I}}_1$, where $\mathcal{I}$ is the identity map acting on $\mathrm{L}(\mathcal{H})$. If $A$ and $B$ are matrices, then we say $A \geq B$ if $A - B$ is \emph{positive semi-definite}. Finally, a \emph{block encoding} of $A$ is a unitary matrix $U_A$ that, for some ancilla-count $k$, satisfies:
\begin{align}
A = \left(\bra*{0^k} \otimes I\right)U_A\left(\ket*{0^k} \otimes I\right).
\end{align}

\subsection{Gibbs states and Davies generators}
We now define some fundamental concepts required to state our quantum algorithm and to analyze its performance.  

%%%

\begin{definition}\label{def:hamiltonian}
Let $H$ be a Hamiltonian on the Hilbert space $\mathcal{H}$ with eigendecomposition 
\begin{align}
    H 
    &= 
    \sum_i \lambda_i \Pi_i.
\end{align}
Here, the $\Pi_i$ are projectors onto the subspace with energy $\lambda_i$. We assume that the spectrum of $H$ is contained in the interval $[0,1]$.  For inverse temperature $\beta>0$, the \textbf{Gibbs state} $\rho_\beta$ is the state such that
\begin{align}
    \rho_\beta &\propto \exp(-\beta H).
\end{align}
The \textbf{partition function} $\mathcal{Z}_\beta$ is the normalization factor given by
\begin{align}
    \mathcal{Z}_\beta &= \mathrm{Tr}\left( \exp(-\beta H) \right).
\end{align}
\end{definition}

%%%

Our quantum algorithm makes it possible to approximately prepare thermal states $\rho_\beta$.  It is based on the Davies generators defined below. The Davies generators describe dissipative dynamics that converge to thermal states.

%%%

\begin{definition}\label{def:daviesgenerator}
  Let $\{S_\alpha\}_\alpha$ be a collection of Hermitian operators acting on $\mathcal{H}$. The \textbf{Davies generator} with respect to the system Hamiltonian $H$ and the \textbf{coupling operators} $S_\alpha$ is the Lindbladian $\mathcal{L}$ in the Schr\"odinger picture given by
\begin{align}
\label{eq:davies}
    \mathcal{L}(\sigma) 
    &= 
    \sum_\omega G(\omega) \left(
    \sum_\alpha S_\alpha(\omega) \sigma S_\alpha(\omega)^{\dagger}
    - 
    \frac{1}{2} \left(
        S_\alpha(\omega)^\dagger S_\alpha(\omega) \sigma + 
        \sigma S_\alpha(\omega)^\dagger S_\alpha(\omega)
    \right) \right).
\end{align}
The \textbf{jump operators} $S_\alpha(\omega)$ are enumerated by the \textbf{Bohr frequencies} $\omega$ of $H$ and are obtained from the coupling operators $S_\alpha$ by
\begin{align}
    S_\alpha(\omega) 
    &= 
    \sum_{\substack{i, j \\ \lambda_i - \lambda_j=\omega}} \Pi_i S_{\alpha} \Pi_j.
\end{align}
The \textbf{filter function} $G(\omega)$ is a real-valued function satisfying $G(\omega)=e^{\beta\omega} G(-\omega)$. 

The time evolution of quantum states is given by the quantum channels
\begin{align}
    \mathcal{T}_t &= \exp(t \mathcal{L}) 
\end{align}
for $t\ge 0$. 

We say that a quantum state $\rho$ is a \textbf{fixed point} of the Davies generator $\mathcal{L}$ if 
\begin{align}
    \mathcal{L}(\rho)=0.
\end{align}
\end{definition}

It is known that if the fixed point of the Davies generator $\mathcal{L}$ is unique, then the Lindbladian time evolution is relaxing in the sense that 
\begin{align}
    \lim_{t\rightarrow\infty} \mathcal{T}_t(\sigma) 
    &=
    \rho_\beta
\end{align}
for any initial state $\sigma$ (see \cite{DN19} and the references therein). The converse clearly holds.

\begin{definition} \label{def:mixing}
We call the coupling operators $S_\alpha$ of a Davis generator $\mathcal{L}$ \textbf{mixing} if the thermal state $\rho_\beta$ is the unique fixed point of $\mathcal{L}$. In this case, $t_{\mathrm{mix}}$ denotes the \textbf{mixing time} of $\mathcal{L}$, that is, the smallest time $t$ such that $\mathcal{T}_t(\sigma)$ is guaranteed to be sufficiently close to the desired thermal state $\rho_\beta$ for any initial state $\sigma$.
\end{definition}

There are several sufficient conditions guaranteeing the uniqueness of the fixed point.  
For instance, it can be shown that the thermal state $\rho_\beta$ is the unique fixed point of the Davis generator $\mathcal{L}$, if the matrix algebra generated by the coupling operators $S_\alpha$ is the full matrix algebra (this statement follows from \cite{Spohn77}; see also the discussion in \cite{DN19} for an overview of other sufficient conditions).

The Lindbladian $\mathcal{L}$ of the Davies generator can also be written in terms of a collection of jump operators $\{L_{\omega,\alpha}\}_{\omega,\alpha}$:
\begin{align} \label{eq:lindblad}
    \mathcal{L}(\sigma) = \sum_{\omega,\alpha} L_{\omega,\alpha} \sigma L_{\omega,\alpha}^\dagger - \frac{1}{2}\left( L_{\omega,\alpha}^\dagger L_{\omega,\alpha} \sigma + \sigma L_{\omega,\alpha}^\dagger L_{\omega,\alpha} \right)
\end{align}
where $L_{\omega,\alpha} = \sqrt{G(\omega)}S_\alpha(\omega)$. Having brought the Davies generator into this form, we can leverage existing results for Lindblad simulation \cite{1612.09512}.

In particular, say we can implement a block encoding of the operator $\mathcal{O}_\mathcal{L}$ that implements the jump operators as follows:
\begin{align}
  \mathcal{O}_\mathcal{L}  (\ket{0} \otimes \ket{\psi}) =  \sum_{\omega,\alpha} \ket{\omega,\alpha} \otimes L_{\omega,\alpha} \ket{\psi}
\end{align}
Then, given access to $\mathcal{O}_\mathcal{L}$, we have technical tools to simulate the time evolution $e^{\mathcal{L}t}$. We use the simulation algorithm from~\cite{LW22}, which simplifies the simulation algorithm of~\cite{1612.09512} and also generalizes their input models to block-encodings.

\begin{proposition}[\textbf{Lindblad simulation, adapted from}~{\cite{LW22}}]
\label{prop:lindbladsimulation}  
Say $\mathcal{L}$ is a Davies generator acting on $k$ qubits with $m$ many jump operators $L_{\omega,\alpha}$ with $\norm{L_{\omega,\alpha}} \leq 1$, and say we are given access to oracles to the jump operators via the block encoding $\mathcal{O}_\mathcal{L}$ above. Then, for any $t \geq 0$ and any $\delta_\mathrm{sim} > 0$, there exists a quantum algorithm that implements a quantum channel $\delta_\mathrm{sim}$-close in the diamond norm to $e^{t\mathcal{L}}$, making 
    \begin{align}
      O\left(mt\,\frac{\log(mt/\delta_{\mathrm{sim}})}{\log\log(mt/\delta_{\mathrm{sim}})}\right)
    \end{align}
    uses of the block encoding of $\mathcal{O}_\mathcal{L}$, and 
    \begin{align}
      O\left(m^2t\,\left(\frac{\log (mt/\delta_{\mathrm{sim}})}{\log\log (mt/\delta_{\mathrm{sim}})}\right)^2\right)
    \end{align}
    additional 1- and 2-qubit gates.
\end{proposition}
In our adaptation of the theorem above from \cite{LW22}, we have replaced a quantity $||\mathcal{L}_\mathrm{be}||$ with a bound $O(m)$ which follows from the fact that $||L_{\alpha,\omega}|| \leq 1$.

So, the central task is to implement a block encoding of $\mathcal{O}_\mathcal{L}$ for the Davies generator, which essentially amounts to implementing block encodings of the $S_\alpha(\omega)$ operators. Once this is accomplished, we can simulate the evolution $e^{t_\text{mix}\mathcal{L}}$ on any input state, and obtain an approximation of the thermal state $\rho_\beta$.

\subsection{Rounding promises}
\label{sec:roundingpromises}

Unfortunately, it is not possible to block encode the jump operators $S_\alpha(\omega)$ of Davies generators without prior knowledge of the spectrum of $H$. Ideally, we would like to implement a unitary that performs the isometry:
\begin{align}
    \sum_{i} \Pi_i \otimes \ket{\lambda_i},
\end{align}
that is, computes a binary representation of the energy $\lambda_i$ into a new register. However, there are two unavoidable limitations. First, $\lambda_i$ can only be estimated to precision $\varepsilon$ using only $O(1/\varepsilon)$ many resources. So, unless we select $\varepsilon$ to be less than the smallest gap between any pair of energies, which requires at least $O(\text{dim}(H))$ many resources, we will be unable to distinguish certain energies. We find that, leveraging a trick from Poulin and Wocjan~\cite{0905.2199}, this error can be dealt with formally.

However, the second limitation is more challenging to deal with. Rall~\cite{2103.09717} observed that for any energy estimation algorithm there will exist particular energies $\lambda_i$ such that a corresponding eigenstate $\ket{\psi_i}$ will produce a superposition of estimates:
\begin{align}
    \ket{\psi_i} \mapsto \ket{\psi_i} \otimes \left( \alpha \ket*{\tilde\lambda_i} + \beta \ket*{\tilde\lambda'_i} \right)
\end{align}
where $\alpha,\beta$ are complex coefficients and $\tilde \lambda_i, \tilde \lambda'_i$ are two estimates of the energy $\lambda_i$. Essentially, certain energies $\lambda_i$ are indecisive about which direction they want to round, and end up being rounded up or down in superposition.

The superposition of rounding directions creates cross terms between the rounding results in the construction of approximate block encodings of the jump operators $S_\alpha(\omega)$, which cause errors in spectral norm of constant size no matter how high the precision of energy estimation, at least using SVT-techniques. A new approach is needed.

To remove the superposition of rounding errors, Rall~\cite{2103.09717} introduced the notion of a `rounding promise'. Observing that certain $\lambda_i$ are indecisive about their rounding, the rounding promise simply asserts that these $\lambda_i$ do not appear in $H$. This assumption on $H$ is very unphysical. In our work, we introduce a new notion of a rounding promise that can be guaranteed without restricting the Hamiltonian, and makes a similar guarantee. The central idea is similar: certain energies $\lambda_i$ are disallowed. But instead of being an assumption on the Hamiltonian itself, our rounding promises define subspaces of the Hilbert space.

\begin{definition} \label{def:roundingpromise}

A \textbf{rounding promise} $M$ is a collection of $s^{(M)}$ many intervals $[a^{(M)}_x,b^{(M)}_x]\subseteq [0,1]$, enumerated by the label $x\in\{0,1,\ldots, s^{(M)}-1\}$, such that  $b^{(M)}_x < a^{(M)}_{x+1}$ for all $x$. We always use the convention that the first interval starts at $0$ and the last one ends at $1$.
A \textbf{gap} of $M$ is a connected open subinterval $(b^{(M)}_x, a^{(M)}_{x+1})$ that spans the gap between two adjacent intervals $[a^{(M)}_x,b^{(M)}_x]$ and $[a^{(M)}_{x+1},b^{(M)}_{x+1}]$. If we write $\lambda \in M$, we mean that $\lambda$ is contained in the union of all the intervals.

The \textbf{promised eigenspace projector} $P^{(M)}_x$ is the projector onto the eigenspaces of $H = \sum_i \lambda_i \Pi_i$ whose eigenvalues lie in the interval $[a^{(M)}_x,b^{(M)}_x]$, that is:
\begin{align}
    P^{(M)}_x 
    &\coloneqq 
    \sum_{\substack{i \\ \lambda_i \in [a^{(M)}_x,b^{(M)}_x]}} \Pi_i.
\end{align}

The \textbf{promised subspace projector} $P^{(M)}$ is the projector onto the eigensubspaces of $H$ whose eigenvalues lie in $M$, that is:
\begin{align}
    P^{(M)} 
    &\coloneqq
    \sum_{x=0}^{s^{(M)}-1} P^{(M)}_x.
\end{align}

The \textbf{promised subspace} $\mathcal{P}^{(M)}$ is 
\begin{align}
    \mathcal{P}^{(M)}
    &\coloneqq 
    \text{image of } P^{(M)}.
\end{align}

\end{definition}

\begin{remark}
When the rounding promise $M$ is fixed, we often omit the superscript $(M)$ to abbreviate the notation.  For instance, we write $\mathcal{P}$ instead of $\mathcal{P}^{(M)}$ for the promised subspace projector.
\end{remark}

Rather than restricting the Hamiltonian itself, we have defined a subspace $\mathcal{P}^{(M)}$ on which energy estimation can be performed without superpositions of rounding errors. Furthermore, the rounding promise also conveniently specifies the estimates themselves: estimating $\lambda$ amounts to identification of the index $x$ such that $\lambda \in [a_x^{(M)}, b_x^{(M)}]$. Indeed, we have the following result:

\begin{proposition} [\textbf{Energy estimation given a rounding promise}~{\cite{2103.09717}}] \label{prop:energyestimationgivenaroundingpromise}   Say $M$ is a rounding promise, say $\kappa$ is the length of the smallest gap. Suppose also that the number of intervals $s^{(M)}$ satisfies $s^{(M)} \leq 2^{n+1}$, so each label can be thought of as a bit string $x \in \{0,1\}^{n+1}$. 

Then, for any $\delta_\mathrm{est} >0$, there exists a family of operators $\tilde P_x$ that approximate $P_x$ in the sense that:
\begin{align}
   \norm{ \tilde P_x P^{(M)} - P_x } \leq \delta_\mathrm{est}.
\end{align}
In fact, $P_x,\tilde P_x$ and $P^{(M)}$ commute. Moreover, say we have a block encoding of a Hamiltonian $H$. Then, for any $\delta_\mathrm{est}$, there exists a quantum circuit that implements an isometry $\tilde E^{(M)}$ satisfying:
\begin{align}
  \tilde E^{(M)} = \sum_x \ket{x} \otimes G_{\tilde P_x}
\end{align}
where the $ G_{\tilde P_x}$ are isometries satisfying $ G_{\tilde P_x}^\dagger G_{\tilde P_x} = \tilde P_x$.
This circuit can be implemented using $O( n^2 \kappa^{-1} \log(\delta^{-1}_\mathrm{est}))$ applications of the block encoding of $H$ and 1- and 2-qubit gates.
\end{proposition}
We prove this in \cref{app:polynomialconstruction}. The factor of $O(n^2)$ can be removed with some additional care using techniques from \cite{2103.09717}, but we keep it here to simplify the algorithm. This factor corresponds to the $O(\log^2(\beta/\eps))$ in the performance in Theorem~\ref{thm:mainresult}.

\section{Algorithm overview}
While we cannot use Davies generators to exactly prepare the Gibbs state $\rho_\beta$ on the entire space, we can prepare the promised Gibbs states $\rho_\beta^{(M)}$ which are restricted to a particular promised subspace $\mathcal{P}^{(M)}$.

\begin{definition} \label{def:promisedgibbs} Let $M$ be a rounding promise. For each $x$, let $m^{(M)}_x$ denote the midpoint of the interval $[a_x, b_x]$. For $\beta \geq 0$, the \textbf{promised Gibbs state} $\rho_\beta^{(M)}$ is the density matrix supported only on the promised subspace $\mathcal{P}^{(M)}$ such that
\begin{align}
    \rho_\beta^{(M)}
    &\propto 
    \sum_{x} 
    \exp \big( -\beta m^{(M)}_x \big) \, P^{(M)}_x.
\end{align}
The \textbf{promised partition function} $\mathcal{Z}_\beta^{(M)}$ is the normalization factor 
\begin{align}
    \mathcal{Z}_\beta^{(M)}
    &\coloneqq 
    \sum_{x} 
    \exp \big( -\beta m^{(M)}_x \big) \, \mathrm{Tr}\big( P^{(M)}_x \big).
\end{align}
\end{definition}

With this idea in place, we can give an informal high-level overview of our algorithm. Each of the major challenges in the algorithm's construction is treated in a section of this paper.

First, promised thermal states are in general not close to the ideal thermal state. However, we find that a suitable ensemble of promised thermal states is satisfactory.
\begin{claim}[\cref{sec:averagingtogetherpromisedthermalstates}]
The desired thermal state $\rho_\beta$ for the Hamiltonian $H$ can be approximated by an average of promised thermal states $\rho_\beta^{(M_j)}$ for suitably chosen rounding promises $M_j$, where $j\in\{0,\ldots,2^r-1\}$.  In particular, let $\rho^{*}$ be the ensemble average over the promised thermal states of the $M_j$. Then, if we perform energy estimation to $n$ bits of precision:
\begin{align}
    \norm{  \rho^{*} - \rho_\beta  }_1 \leq \sqrt{\beta \cdot 2^{-n}} + 2\cdot 2^{-r} 
\end{align}
\end{claim}

The main idea is that the rounding promises $M_j$ have to be chosen such that each eigenvalue $\lambda_i$ of the Hamiltonian $H$ is contained in at least $2^r-1$ rounding promises $M_j$. That way any individual eigenspace can be missing with probability at most $2^{-r}$.

Second, we must restrict the dynamics of a Davies generator to a promised subspace. This `promised Davies generator' will then be used to prepare the promised thermal states.

\begin{definition}[\textbf{Promised Davies generator}] \label{def:promiseddaviesgenerator} Say $\mathcal{L}$ is a Davies generator on the full Hilbert space with coupling operators $S_\alpha$. Say $M$ is a rounding promise, and we are given an `attenuation operator' $A^{(M)}$ that commutes with the Hamiltonian and satisfies:
\begin{align}
    \norm{A^{(M)} (I - P^{(M)})} = 0.
\end{align}

Then, the \textbf{promised Davies generator} $\mathcal{L}^{(M)}$ is a Lindblad operator defined by the jump operators $L_{\nu,\alpha}$ given by:
\begin{align}
    L^{(M)}_{\nu, \alpha} = \sqrt{G(\nu)} S^{(M)}_{\alpha}(\nu).
\end{align}

\begin{comment}
\begin{align}
    \mathcal{L}^{(M)}(\rho)
    &=
    \sum_{\nu} G(\nu) 
    \sum_\alpha \left[
    S_\alpha^{(M)}(\nu) \rho S_\alpha^{(M)}(\nu)^\dagger
    -
    \frac{1}{2} \Big(
    S_\alpha^{(M)}(\nu)^\dagger S_\alpha^{(M)}(\nu) \rho
    + 
    \rho S_\alpha^{(M)}(\nu)^\dagger S_\alpha^{(M)}(\nu)
    \Big)
    \right].
\end{align} 
\end{comment}

Its Bohr frequencies $\nu$ are differences of the form $m^{(M)}_x - m^{(M)}_y$, where $m^{(M)}_x$ and $m^{(M)}_y$ are the midpoints of the intervals $[a_x, b_y]$ and $[a_y,b_y] \in M$, respectively. Its jump operators $S_\alpha^{(M)}(\nu)$ are given by
\begin{align}
    S_\alpha^{(M)}(\nu) 
    &=
    \sum_{\substack{x, y \\ m^{(M)}_x - m^{(M)}_y = \nu}} P_x S_\alpha^{(M)} P_y.
\end{align}

Finally, the coupling operators $S_\alpha^{(M)}$  are given by
\begin{align}
    S_\alpha^{(M)} 
    &= 
    A^{(M)} S_\alpha A^{(M)}.
\end{align}
\end{definition}

In Section~4 we will construct the attenuation operator $A^{(M)}$ as well as show that the promised Davies generator has the following properties:

\begin{claim}[\cref{sec:lindbladdynamicsonthepromisedsubspace}] Say $\mathcal{L}$ is a Davies generator on the full Hilbert space with coupling operators $S_\alpha$. Then, for any rounding promise $M$, there exists a \textbf{promised Davies generator} $\mathcal{L}^{(M)}$ with the following properties:
\begin{itemize}
    \item  If an input state $\sigma$ is supported only on $\mathcal{P}^{(M)}$, the output state $e^{t\mathcal{L}^{(M)}}(\sigma)$ will be as well.
    \item The unique fixed point of $\mathcal{L}^{(M)}$ is the promised thermal state $\rho^{(M)}_\beta$.
    \item Numerical simulations indicate that the mixing time $\tilde t_\text{mix}$ of $\mathcal{L}^{(M)}$ is not too much slower than that of $\mathcal{L}$.
\end{itemize}
\end{claim}
 
The purpose of the attenuation operator $A^{(M)}$ is to ensure that the output of the Lindblad evolution according to $\mathcal{L}^{(M)}$ is confined to the promised subspace. Then, so long as time evolution starts in an initial state $\tilde{\sigma}^{(M)}$ that is approximately supported on the promised subspace $\mathcal{P}^{(M)}$, then the output state will be as well.

The attenuation operator can make the mixing time of the Lindblad evolution slower. First, even in the ideal case when $A^{(M)} = P^{(M)}$, the elimination of certain energy eigenspaces may result in slower mixing. Second, it is not possible to project onto the promised subspace perfectly, for similar reasons that it is not possible to estimate energies without superpositions of rounding errors. We solve this problem by attenuating some of the eigenspaces in the promised subspace as well. This may also result in slower mixing. We show via numerical study that neither of these effects make the mixing time too much worse in practice.

\begin{remark}[\textbf{Well-rounded Hamiltonian}] \label{remark:wellroundedhamiltonian}
Note that the promised Davies generator $\mathcal{L}^{(M)}$ is block diagonal with respect to the orthogonal decomposition $\mathcal{H}=\mathcal{P}^{(M)}\oplus{\mathcal{P}^{(M)}}^\perp$.  It acts as $0$ on ${\mathcal{P}^{(M)}}^\perp$, and acts on $\mathcal{P}^{(M)}$ as the Davies generator with respect to the promised Hamiltonian $H^{(M)}\coloneqq\sum_x m^{(M)}_x P^{(M)}_x$ and coupling operators $S_\alpha^{(M)}$ (where we view all operators as restricted to the promised subspace $\mathcal{P}^{(M)}$).
\end{remark}

Third and finally, we show how to construct an approximate block encoding $O_{\mathcal{L}^{(M)}}$ for each promised Davies generator. This lets us leverage \cref{prop:lindbladsimulation} to simulate the Lindblad dynamics and prepare the promised thermal states.

\begin{claim}[\cref{sec:implementinglindbladdynamics}] \label{claim:implementingdynamics} Say we are given block encodings of some coupling operators $S_{\alpha}$, and one of the rounding promises $M \in \{M_0, \ldots,M_{2^{r}-1}\}$. Suppose we perform energy estimation to $n$ bits of precision. Then, there exists a quantum algorithm that implements time evolution for time $t$ according to $\mathcal{L}^{(M)}$ to any precision $\delta_\mathcal{L}>0$ in the diamond norm using:
\begin{align}
  O\left( \gamma^{-1} \cdot n^2 2^{3n+r} t \cdot \mathrm{polylog}(t/\delta_{\mathcal{L}})\right)
\end{align}
invocations to the block encoding of $H$, where $\gamma$ is an attenuation coefficient introduced in \cref{sec:lindbladdynamicsonthepromisedsubspace}.
We accomplish this by constructing and simulating time evolution according to an approximate promised Davies generator $\tilde{\mathcal{L}}^{(M)}$. 

\end{claim}

\section{Averaging together promised thermal states\label{sec:averagingtogetherpromisedthermalstates}}

By selecting a rounding promise $M$ restricting to a promised subspace $\mathcal{P}^{(M)}$, we have gained the ability to approximately prepare promised Gibbs states ${\rho_\beta}^{(M)}$. These Gibbs states have support only on $\mathcal{P}^{(M)}$, and thus may be far from the true thermal state $\rho_\beta$. In this section we show how to construct several rounding promises $M_j \in \{M_0,\ldots,M_{2^{r}-1}\}$, such that the ensemble average over the ${\rho_\beta}^{(M_j)}$'s is provably close to $\rho_\beta$.

Furthermore, to prepare an approximation of ${\rho_\beta}^{(M_j)}$, we require an initial state $\sigma$ that is also approximately supported only on $\mathcal{P}^{(M_j)}$. This state is then fed as input to the promised Davies generator, whose dynamics are trivial outside of $\mathcal{P}^{(M_j)}$. We also show how to prepare these initial states in this section: it is achieved by measuring a POVM called the `left-right POVM' on an arbitrary initial state.

The left-right POVM splits the computation into two branches, each corresponding to a family of left rounding promises $\bar L, L_j$ and right rounding promises $\bar R, R_j$. We use $M$ as a symbol to denote either $L$ or $R$ depending on which branch we are in. An overview of the algorithm as a whole is as follows:
\begin{description}
    \item[Step 1.] Take an arbitrary initial state $\sigma$, and measure the left-right POVM defined in \cref{sec:leftrightpovm}. Depending on the measurement outcome, the resulting state $\tilde{\sigma}^{(M)}$ will be approximately supported on $\mathcal{P}^{(\bar M)}$, where $\bar M \in \{\bar L,\bar R\}$ is one of two `fine-grained' rounding promises.

    \item[Step 2.] Pick an index $j \in \{0,...,2^{r}-1\}$ uniformly at random. This index determines a `coarse-grained' rounding promise $M_j$, which is either $L_j$ or $R_j$ depending on the measurement outcome of the left-right POVM. These are defined in \cref{sec:coarsegrainedroundingpromises}. Since the $M_j$ are coarse-grainings of $\bar M$, the input state $\tilde{\sigma}^{(M)}$ is supported only on $\mathcal{P}^{(M_j)}$ for any $j$.
    \item[Step 3.] Use the promised Davies generator to approximately prepare ${\rho_\beta}^{(M_j)}$. This is discussed in \cref{sec:lindbladdynamicsonthepromisedsubspace,sec:implementinglindbladdynamics}.
    \item[Analysis.] Let $\rho^{*M}$ be the ensemble average over the index $j$ that we selected in step 2 and used to prepare ${\rho_\beta}^{(M_j)}$. We show that for both $M \in \{L,R\}$, the density matrix $\rho^{*M}$ is close in trace distance to the ideal thermal state $\rho_\beta$. We perform this analysis in \cref{sec:ensembleanalysis}. 
\end{description}

This protocol is represented diagrammatically in \cref{fig:protocoldiagram}.
\begin{figure}
    \centering
    \includegraphics[width=0.8\textwidth]{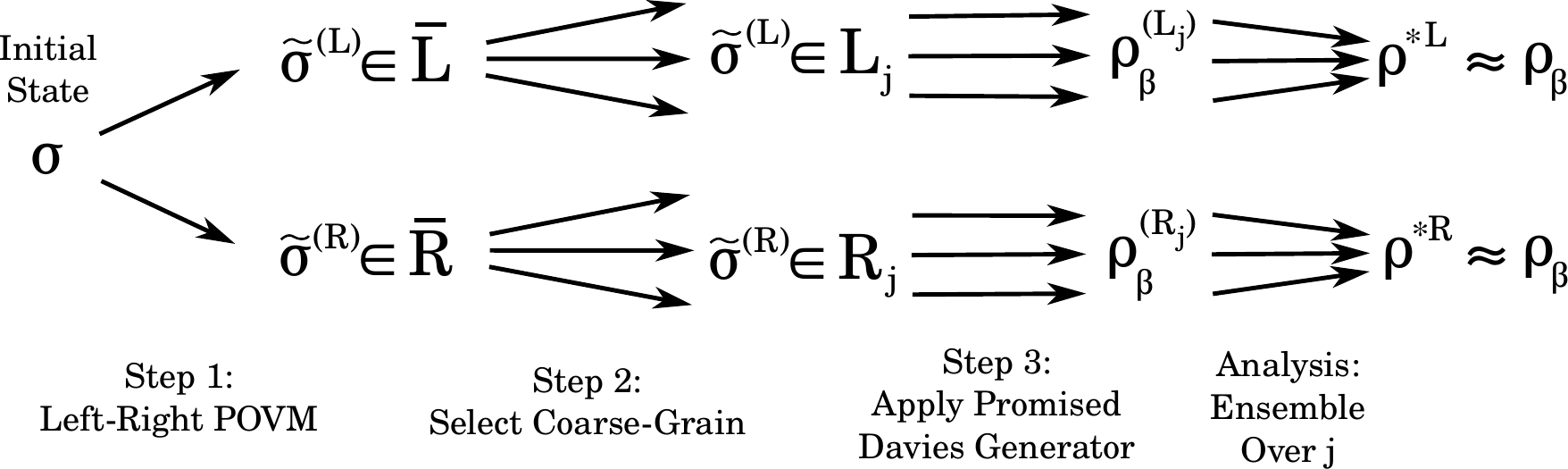}
    \caption{\label{fig:protocoldiagram} Sketch of a protocol that leverages promised Davies generators to prepare an approximation $\rho^{*M}$ of an ideal thermal state $\rho_\beta$ (for $M \in \{L,R\}$). In the figure above, $\rho \in \bar M$ is a shorthand for `$\rho$ is approximately supported entirely on $\mathcal{P}^{(\bar M)}$', which we define more rigorously in \cref{def:approximatesupport}.}
\end{figure}

\subsection{The Left-Right POVM\label{sec:leftrightpovm}}

The Left-Right POVM has the purpose of producing a quantum state $\tilde \sigma^{(M)}$ which is guaranteed to be approximately supported on one of the two fine-grained rounding promises $\bar M \in \{\bar L, \bar R\}$ depending on the measurement outcome. In this section we show how to implement this POVM, and prove that the post-measurement state has the desired property. We define this notion rigorously now:

\begin{definition} \label{def:approximatesupport} Say $M$ is a rounding promise and $\tilde{\sigma}$ a density matrix. We say $\tilde\sigma$ is  $\delta_\mathrm{sup}$\textbf{-approximately supported on} $M$ if
\begin{align}
  \Tr( \tilde\sigma P^{(M)} ) \geq 1-\delta_\mathrm{sup}.
\end{align}
In other words, if we measure the projector $P^{(M)}$ in order to project into the promised subspace $\mathcal{P}^{(M)}$, we succeed with probability $\geq 1-\delta_\mathrm{sup}$.
\end{definition}

The construction depends on two parameters: $n$ and $r$. Here $n$ is the number of bits of precision for energy estimation, and $r$ determines the number of coarse-grained rounding promises, which is $2^r$. The implementation of the left-right POVM as well as energy measurement using \cref{prop:energyestimationgivenaroundingpromise} will require $O(2^{n+r})$ applications of the block encoding of $H$. This is because all the rounding promises in this construction have a minimum gap size $\kappa = 2^{-n-r-2}$. 

The quantities $2^{-n}$ and $2^{-r}$ correspond to two different errors on the final state. The quantity $2^{-n}$ corresponds to the accuracy of energy estimation, and $2^{-r}$ is the probability with which any particular energy is excluded from the ensemble. We will show in \cref{thm:accuracyofthefinalensemble} that if output state is $\rho^*$, then we have the guarantee is $\norm{\rho^* - \rho_\beta}_1 \leq \sqrt{\beta 2^{-n}} + 2\cdot 2^{-r}$. In our final construction in \cref{sec:implementinglindbladdynamics} we will select $2^n \sim \beta/\eps^{2}$ and $2^r \sim 1/\eps$, achieving an accuracy of $\sim\eps$.

As we present the construction, we recommend following along with \cref{fig:roundingpromises}. The main idea is that we would like to eliminate small regions of the spectrum via the Left-Right POVM, which is defined by an operator $P_\mathrm{LR}$ whose spectrum is sketched in \cref{fig:roundingpromises}. When $P_\mathrm{LR}$ has no support on an eigenstate, and we observe the POVM outcome corresponding to $P_\mathrm{LR}$, then the output state has no support on that eigenstate. Consequently the output state satisfies the rounding promise $\bar L$. Similarly, the rounding promise $\bar R$ is defined by the eigenstates on which $P_\mathrm{LR}$ has eigenvalue 1. The key property of $\bar L$ and $\bar R$ that we require for the remainder of the construction is that there are $2^{n+r}$ many evenly spaced gaps in the spectrum. Later, we will define coarse-grained rounding promises that close all but $2^n$ of these gaps at random, so that the probability of any individual energy being excluded by the coarse-grained promise is at most $2^{-r}$. 

\begin{figure}
    \centering
    \includegraphics[width=0.9\textwidth]{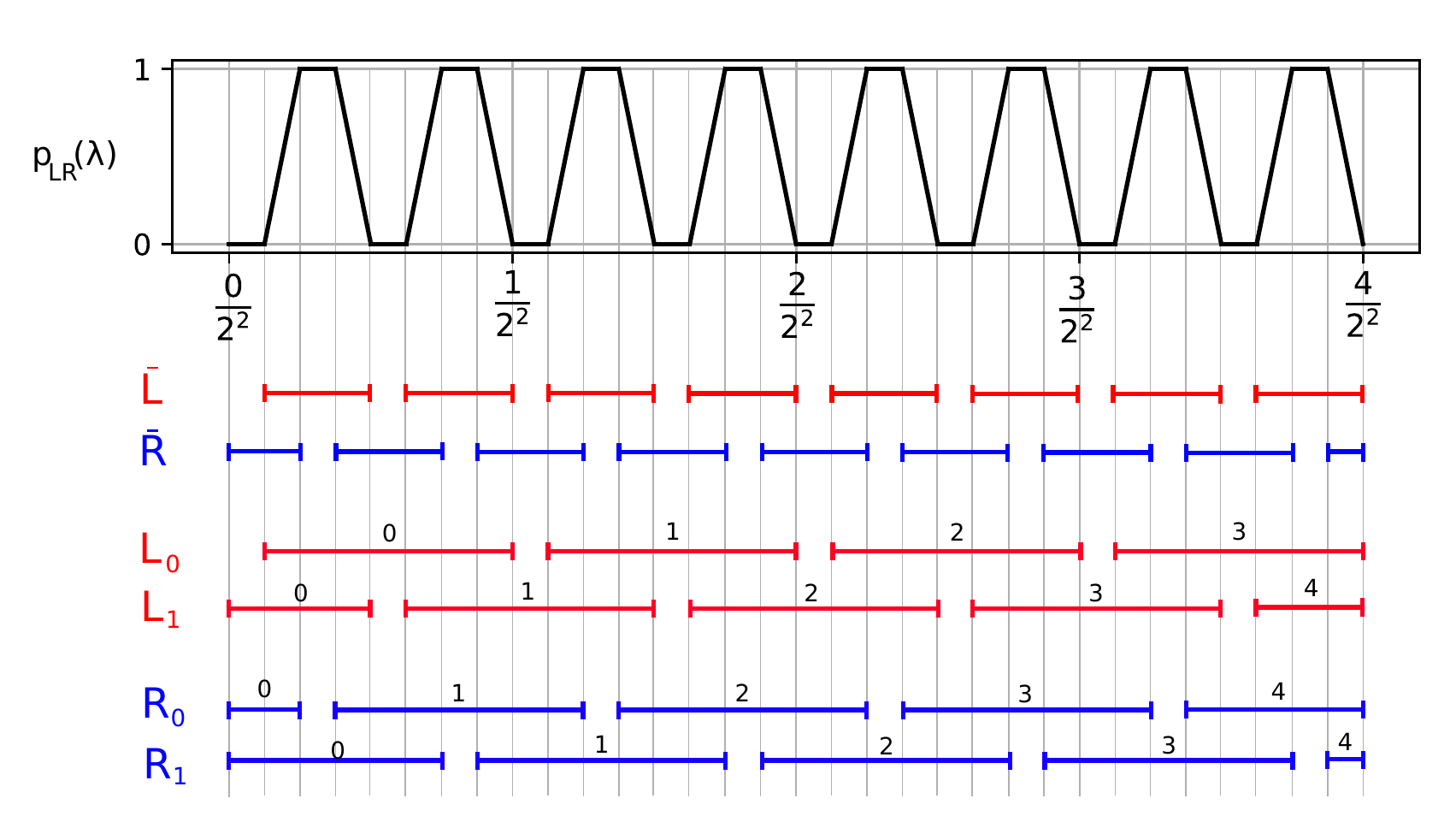}
    \caption{\label{fig:roundingpromises} Sketch of the rounding promises from this section with $n=2$, and $r=1$. The diagram features a sketch of the function $p_{LR}(\lambda)$ that defines the projector $P_{LR}$ from \cref{lemma:leftrightprojectionoperator}, a sketch of the fine-grained rounding promises from \cref{def:thefinegrainedroundingpromises}, and a sketch of the coarse-grained rounding promises from \cref{def:thecoarsegrainedroundingpromises}. The numbers above the intervals in the coarse-grained promises indicate the index $x$. }
\end{figure}

\begin{definition}[\textbf{The fine-grained rounding promises}]
\label{def:thefinegrainedroundingpromises} 
For $n,r \in \mathbb{Z}_+$, let the rounding promises $\bar L$ and $\bar R$ be obtained by taking the full interval $[0,1]$ and removing some subintervals $(a',b')$. Specifically: % Every time we delete a subinterval, we take the unique $[a_x,b_x]$ satisfying $(a',b')\subset [a_x,b_x]$, and replace it with two intervals $[a_x,a'],[b',b_x]$ (if $a_x = a'$, just replace $[a_x,b_x]$ with $[b',b_x]$).
\begin{align}
  \text{To construct }\bar L,\text{ delete }  \left( \frac{k}{2^{n+r+2}}, \frac{k+1}{2^{n+r+2}},  \right)\text{ for } k \in \{0,\ldots, 2^{n+r}-1\}. \\
      \text{To construct }\bar R,\text{ delete } \left( \frac{k+2}{2^{n+r+2}}, \frac{k+3}{2^{n+r+2}},  \right) \text{ for } k \in \{0,\ldots, 2^{n+r}-1\}.
\end{align}
Let the $[a_x,b_x]$ denote the resulting connected components.  These closed intervals define the rounding promise in $\bar L$ and $\bar R$ respectively.
\end{definition}

To prepare a state supported only on $\bar L$ or $\bar R$, we need to take our input state $\sigma$ and remove all the support on the eigenspaces in the deleted regions above.  To this end, we construct a block encoding of an operator $P_\mathrm{LR}$ that makes $P_\mathrm{LR}\sigma P_\mathrm{LR}$ be supported only on $\bar L$ and $(I-P_\mathrm{LR})\sigma (I-P_\mathrm{LR})$ be supported only on $\bar R$.

\begin{lemma}[\textbf{Left-right projection operator}]
\label{lemma:leftrightprojectionoperator} For any $n,r \in \mathbb{Z}_+$ and any $\delta_\mathrm{sup} > 0$, there exists a Hermitian operator $P_{\mathrm{LR}}$ that satisfies $\norm{P_\mathrm{LR}^2\Pi_i} \leq \delta_\mathrm{sup}/3 $ for $\lambda_i \not\in \bar L$, and satisfies  $\norm{(I-P_\mathrm{LR}^2)\Pi_i} \leq \delta_\mathrm{sup}/3 $ for $\lambda_i \not\in \bar R$. This operator has a block encoding that can be implemented using $O((n+r) 2^{n+r} \log(\delta_\mathrm{sup}^{-1}) )$ invocations of a block encoding of $H$.
\end{lemma}
\begin{proof} See \cref{app:polynomialconstruction}. We will ensure $P_\mathrm{LR}$ commutes with the Hamiltonian, and takes the form $P_\mathrm{LR} = \sum_i p_\mathrm{LR}(\lambda_i) \Pi_i$ for a function $p_\mathrm{LR}$ that is plotted in \cref{fig:roundingpromises}.
\end{proof}

It remains to show how to turn $P_\mathrm{LR}$ into a POVM, and how to implement that POVM using a quantum circuit. The starting point for the implementation is the block-measurement theorem from \cite{2103.09717}. This theorem takes a block encoding of an operator $A$ satisfying $|A^2 - \Pi| \leq \delta$ for some projector $\Pi$, and uses it to implement a quantum channel $4\sqrt{2}\delta$-close in $\diamond$-norm to the isometry:
\begin{align}
    \ket{1} \otimes \Pi + \ket{0} \otimes (I - \Pi).
\end{align}
This does not quite suffice for us, since $P_\mathrm{LR}$ is not a projector. But, if we allow ourselves to introduce a constant amount of postselection, then we can implement a very similar operation that suffices for our purposes.

\begin{lemma}[\textbf{Postselective block-measurement}] \label{lemma:postselectiveblockmeasurement} Say $A$ is a Hermitian matrix with a block encoding $U_A$. Then there exists a quantum circuit with postselection using one $U_A$ and one $U_A^\dagger$ that implements the map $\sigma \to V_A \sigma V_A^\dagger / \text{Tr}(V_A \sigma V_A^\dagger) $, where:
\begin{align}
   V_{A} \coloneqq \ket{1} \otimes A^2 + \ket{0} \otimes (I-A^2).
\end{align}
The postselection succeeds with probability at least $1/2$.
\end{lemma}

\begin{proof} Say the block encoding $U_A$ has $k$ ancillae, that is $\left(\bra{0^k}\otimes I\right)U_A\left(\ket{0^k}\otimes I\right) = A$. We consider the following circuit:
\begin{align}
    \Qcircuit @C=1em @R=.7em {
\lstick{\ket{0}}& \qw & \targ & \qw & \qw \\
\lstick{\ket{0^k}}& \multigate{1}{U_A} & \ctrl{-1} & \multigate{1}{U_A^\dagger} & \measuretab{\text{postselect}\ket{0^k}} \\
& \ghost{U_A} & \qw & \ghost{U_A^\dagger} & \qw
}
\end{align}
where the CNOT above denotes the operator $X \otimes \ketbra*{0^k}{0^k} + I \otimes (I - \ketbra*{0^k}{0^k})$. When the postselection succeeds, this implements $V_A$ as desired. It remains to bound the postselection probability, which is:
\begin{align}
     \text{Tr}(V_A \sigma V_A^\dagger) &= \text{Tr}(A^2\sigma A^2 + (I-A^2)\sigma(I-A^2))\\
     &=\text{Tr}( \sigma \cdot ( A^4 + (I-A^2)^2 )).
\end{align}
Since $A$ is Hermitian and $x^4 + (1 - x^2)^2 \geq 1/2$, we have $A^4 + (I-A^2)^2 \geq I/2$. So we succeed with probability $\geq 1/2$.
\end{proof}

If we measure the top qubit after applying $V_A$, we will implement a POVM defined by the operator $A$.

All the tools are in place to define and implement the left-right POVM and demonstrate that it guarantees that the post-measurement state $\tilde{\sigma}^{(M)}$ has the desired property of being approximately supported on $\bar M$. One caveat is that the procedure only works if the probability of the observed outcome is bounded away from 0 by a constant. This is because the normalization of the state that occurs post-measurement may blow up the support outside the promised subspace. Fortunately, such a lower bound on the outcome probability is easily attained as we will see later.

\begin{proposition}[\textbf{Left-right POVM}]  \label{prop:leftrightpovm}
For any $n,r \in \mathbb{Z}_+$ and any $\delta_\mathrm{sup} > 0$, there exists a two-outcome POVM with the following property. We label the two outcomes `$L$' and `$R$', and for any input state $\sigma$ let the output state be called $\tilde\sigma^{(M)}$ for $M \in \{L,R\}$. Say $p_M$ is the probability of the $M$ outcome, and suppose we obtain an outcome with $p_M \geq 1/3$. Then the output state $\tilde\sigma^{(M)}$ is $\delta_\mathrm{sup}$-approximately supported on $\bar M$.

  Furthermore, there exists a quantum circuit that implements the POVM with success probability $\geq 1/2$ using $O(2^{n+r})$ invocations of a block encoding of $H$.
\end{proposition}
\begin{proof} The circuit is just an invocation of postselective block-measurement from \cref{lemma:postselectiveblockmeasurement} with the left-right projection operator $P_\mathrm{LR}$ from \cref{lemma:leftrightprojectionoperator}. After successfully applying the isometry $V_{P_\mathrm{LR}}$ we measure the label qubit and output `$L$' in the $0$ case and `$R$' in the $1$ case. This immediately gives the bound on the circuit complexity and success probability. It remains to show that the output state has the desired property.

The probabilities of the measurement outcomes are:
\begin{align}
    p_L \coloneqq \text{Tr}(P_\mathrm{LR}^2 \sigma P_\mathrm{LR}^2) \hspace{1cm} \text{and} \hspace{1cm} p_R \coloneqq \text{Tr}((I-P_\mathrm{LR}^2) \sigma (I-P_\mathrm{LR}^2).
\end{align}
The corresponding output states are:
\begin{align}
   \tilde\sigma^{(L)}  \coloneqq P_\mathrm{LR}^2 \sigma P_\mathrm{LR}^2 / p_L \hspace{1cm}\text{and}\hspace{1cm} \tilde\sigma^{(R)}  \coloneqq  (1-P_\mathrm{LR}^2) \sigma (I-P_\mathrm{LR}^2) / p_R.
\end{align}

Another way of writing the results of \cref{lemma:leftrightprojectionoperator} is:
\begin{align}
    \norm{ P_\mathrm{LR}^2  (I-P^{(\bar L)})} \leq  \delta_\text{sup}/3 \text{ and }  \norm{ (I-P_\mathrm{LR}^2)  (I-P^{(\bar R)})} \leq  \delta_\text{sup}/3.
\end{align}
Having set the scene, we are ready the compute the probability of the output register being approximately supported on $\mathcal{P}^{(\bar M)}$. We use the inequality $\text{Tr}(XY)\le\text{Tr}(X)\text{Tr}(Y)$ that holds for all positive semidefinite matrices $X$ and $Y$. We obtain:
\begin{align}
    \text{Tr}( p_L \tilde\sigma^{(L)} \cdot ( I - P^{(\bar L)}) ) &= \text{Tr}(P_\mathrm{LR}^2 \sigma P_\mathrm{LR}^2   \cdot ( I - P^{(\bar L)}) ) \\ 
    &= \text{Tr}(\sigma \cdot  P_\mathrm{LR}^2  (I- P^{(\bar L)}) P_\mathrm{LR}^2  ) \\
    &\leq \text{Tr}(\sigma) \cdot (\delta_\text{sup}/3)^2 \leq \delta_\text{sup}/3 \quad\text{and}\\
        \text{Tr}( p_R \tilde\sigma^{(R)} \cdot ( I - P^{(\bar R)}) ) &= \text{Tr}((I-P_\mathrm{LR}^2) \sigma (I-P_\mathrm{LR}^2)   \cdot ( I - P^{(\bar R)}) ) \\ 
        &= \text{Tr}(\sigma \cdot  (I-P_\mathrm{LR}^2 ) (I- P^{(\bar R)}) (I-P_\mathrm{LR}^2 ) ) \\
        &\leq \text{Tr}(\sigma) \cdot (\delta_\text{sup}/3)^2 \leq \delta_\text{sup}/3.
\end{align}
So for $M \in \{L,R\}$, we have:
\begin{align}
   p_M \cdot \text{Tr}\left( \tilde{\sigma}^{(M)} P^{(\bar M)}  \right) \leq \delta_\text{sup}/3.
\end{align}
Since we assumed $p_M \geq 1/3$, we have $\text{Tr}\left( \tilde{\sigma}^{(M)} P^{(\bar M)}  \right) \leq \delta_\text{sup}$ as desired.

\end{proof}

Achieving $p_M \geq 1/3$ is very easy: we just repeat the protocol a couple of times and pick the outcome we saw the most frequently. 

\begin{corollary}[\textbf{Selecting a high-probability outcome}] For any $\delta_\mathrm{fail},\delta_\mathrm{sup}>0$ there exists a procedure that produces a quantum state $\tilde \sigma^{(M)}$ that, for some random $\bar M \in \{\bar L, \bar R\}$, is $\delta_\text{sup}$-approximately supported on $\bar M$ with failure probability at most $\delta_\mathrm{fail}$. Here, `failure' refers to selecting an outcome $\bar M$ with $p_M < 1/3$. This procedure requires at most $20\cdot \log(\delta_\mathrm{fail}^{-1})$ many re-preparations of the input state $\sigma$ and implementations of the left-right POVM from \cref{prop:leftrightpovm}.
\end{corollary}
\begin{proof} The procedure is as follows. Let $N$ be the smallest odd number greater than $20 \log(\delta_\mathrm{fail})$. We repeat the left-right POVM from \cref{prop:leftrightpovm} a total of $N$ times, preparing a new initial state $\sigma$ each time. We let $M$ be the most frequently observed outcome, and return the output state of any of our attempts that measured the $M$ result.

We fail if $p_M < 1/3$. Let $K$ be the number of times we observed the $M$ outcome. Then, using the Chernoff-Hoeffding theorem, the probability that we observe $K/N > 1/2$ despite the fact that $p_M < 1/3$ is:
\begin{align}
    \mathrm{Pr}[ K/N > 1/2 ] \leq \mathrm{Pr}[  K/N > \mathbb{E}[K]/N + 1/6 ] \leq \text{exp}( 2N (1/6)^2 ) \leq \delta_\mathrm{fail}.    
\end{align}
\end{proof}

\subsection{Coarse-grained rounding promises\label{sec:coarsegrainedroundingpromises}}

After having obtained a state $\tilde{\sigma}^{(M)}$ that (on average) is approximately supported on a fine-grained rounding promise $\bar M$, the next step of the algorithm is to select a coarse-grained rounding promise $M_j \in \{ M_0, ..., M_{2^r-1} \}$. Because $\bar M \subset M_j$ for all $M_j$, we automatically have that $\tilde{\sigma}^{(M)}$ is approximately supported on each of the $M_j$.

The purpose of the coarse-grained rounding promises $M_j$ is to ensure that the ensemble average $\rho^{*M}$ of the promised thermal states ${\rho_\beta}^{(M_j)}$ is close to the ideal thermal state $\rho_\beta$. This would not be true if we just prepared the two states ${\rho_\beta}^{(\bar M)}$ instead and considered their ensemble average, as we will explain now.

Recall that the main mechanism of the rounding promise is to eliminate certain energy eigenspaces. On the one hand, this enables energy estimation without superpositions of rounding errors. On the other hand, since certain energies are eliminated, an individual promised thermal state could never guaranteed be close in trace norm to $\rho_\beta$.

In an ensemble average $\rho^{*M}$ over promised thermal states ${\rho_\beta}^{(M_j)}$ however, we can guarantee that the probability of each individual energy being eliminated is small ($\leq 2^{-r}$). We achieve this condition by selecting each of the $M_j$ with probability $1/2^r$, and showing each energy is eliminated by at most one of the $M_j$'s. We will show in the next section that this suffices to guarantee that $\norm*{ \rho^{*M} - \rho_\beta } \leq \sqrt{\beta 2^{-n}} + 2\cdot 2^{-r}$.

An ensemble over $\rho_\beta^{(\bar M)}$ does not have this property: if an energy $\lambda_i$ is eliminated by, say, $\bar L$, then the probability of the eigenspace being eliminated in the ensemble is $p_L$. If we leave the initial state arbitrary, then $p_L$ could be as large as 1. But even if we set $\sigma$ to the maximally mixed state, thereby ensuring $p_L = p_R = 1/2$, the probability of eliminating the energy is still constant, resulting in a constant-size error. To guarantee an error of at most $\varepsilon$, we require $O(\varepsilon^{-1})$ many rounding promises. We will show that they have width at most $2^{-n}$, so the error from energy rounding will be at most $\sim 2^{-n}$.

We require one more property of the $M_j$'s: recall from \cref{def:promisedgibbs} that the promised thermal state's eigenvalues are $\propto \exp(-\beta m_x^{(M_j)})$, where $m_x$ is the midpoint of the interval $[a_x,b_x]$. This is because when we perform energy estimation according to \cref{prop:energyestimationgivenaroundingpromise}, all of the energies in the interval $[a_x,b_x]$ are rounded to $m_x$. To avoid introducing too much error here, we must ensure that the intervals are not too wide.

The construction of the $M_j$ is also visualized in \cref{fig:roundingpromises}. The intuition for the construction is as follows. We want to achieve $\bar M \subset M_j$, so we construct the $M_j$ my merging gaps between the intervals of $\bar M$. As discussed above, the widths of the intervals in $\bar M$ are $\sim 2^{-n-r}$, and we want the widths of the intervals of $M_j$ to be $\sim 2^{-n}$. This can be achieved by merging all but every $2^r$'th gap. The promise $M_j$ is defined by keeping every gap with index $ = j$ modulo $2^r$.

\begin{definition}[\textbf{The coarse-grained rounding promises}]\label{def:thecoarsegrainedroundingpromises} For $n,r \in \mathbb{Z}_+$ and $M \in \{L,R\}$, let the rounding promises $M_j \in \{ M_0, ..., M_{2^r-1} \}$ be constructed by taking the fine-grained rounding promise $\bar M$ from \cref{def:thefinegrainedroundingpromises} and merging some of the gaps.

By merging the $x$'th gap, we mean the following transformation on a rounding promise: if the promise contains the intervals $[a_x,b_x]$ and $[a_{x+1},b_{x+1}]$, then we take any intervals $[a',b_x],[a_{x+1},b']$ for some $a',b'$, and replace them with $[a',b']$. If $[a',b_{x}]$ is the final interval, then we replace it with $[a',1]$. Finally, if $x = 0$, then we take the first interval $[a',b']$ and replace it with $[0,b']$.

Then, $M_j$ is obtained by starting with $\bar M$ and merging all the $x$'th gaps where $x \neq j \text{ mod } 2^{n+r}$.
\end{definition}

\begin{lemma}[\textbf{Properties of the coarse-grained rounding promises}]
\label{lemma:propertiesofthecoarsegrainedroundingpromises}  First, if $\tilde{\sigma}^{(M)}$ is $\delta_\mathrm{sup}$-approximately supported on $\bar M$, then it is also approximately supported on all the $M_j$'s.  

Second, say $\lambda \in [0,1] $ is any energy. Then at most one of the $M_j$'s does not contain it. Here is another way of saying this:  let $1_{\lambda \in M_j}$ be 1 if $\lambda \in M_j$ and 0 otherwise. Then:
        \begin{align}
            2^{-r} \sum_{j=0}^{2^r-1} 1_{\lambda \not\in M_j} \leq 2^{-r}.
        \end{align}
Finally, each interval $[a_x,b_x]\in M_j$ satisfies $b_x - a_x \leq 2^{-n}$.
\end{lemma}
\begin{proof} For the first claim, since the gap merging procedure ensures that $\bar M \subset M_j$, we have that $P^{(\bar M)} \le P^{(M_j)}$. Consequently, following \cref{def:approximatesupport}:
\begin{align}
    \text{Tr}( \tilde{\sigma}^{(M)} P^{(M_j)}  ) \ge 
    \text{Tr}( \tilde{\sigma}^{(M)} P^{(\bar M)} ) \geq 1 - \delta_\text{sup}. 
\end{align}

For the second claim, we first consider what would happen if we merged \emph{all} the gaps in $\bar M$: since the gap-merging procedure includes cases for the first and last interval, this would just result in the promise $\{[0,1]\}$. So that means that for each $M_j$ the only energies that are missing are those present in the gaps $(b_x,a_{x+1})$ where $x = j \text{ mod } 2^{n+r}$. For every gap index $x$, there is a unique $j$ for which this equation holds. 

For every $\lambda$ there are two cases. If we have $\lambda \in \bar M $, then since $\bar M \subset M_j$ we also have $\lambda \in M_j$ for all $M_j$. Otherwise, since $\lambda\not\in \bar M$, $\lambda$ must be contained in some gap of $\bar M$. By the above, that gap is contained in all but one of the $M_j$, so $\lambda$ must be contained in all but one of the $M_j$'s.

Finally, we must show that the intervals in the $M_j$ are not too wide. The intervals in $\bar M$ all have width $3/4 \cdot 2^{-n-r}$, and are $1/4\cdot 2^{-n-r}$ apart. If we close all but every $2^{r}$'th gap, then the intervals will have width at most:
\begin{align}
    2^r \cdot 3/4 \cdot 2^{-n-r} + (2^r - 1) \cdot 1/4 \cdot 2^{-n-r} \leq  (3/4 + 1/4) \cdot 2^{-n}  \leq 2^{-n}.
\end{align} 

\end{proof}

\subsection{Ensemble analysis\label{sec:ensembleanalysis}}

In this subsection we prove that the ensembles $\rho^{*M}$ are close to the ideal thermal state $\rho_\beta$.

The error comes from two sources: from the accuracy of energy estimation (via $n$), and from number of rounding promises (via $r$). We first deal with the error from the accuracy of energy estimation by defining a new notion of an `exact' promised thermal state ${\hat \rho_\beta}^{(M_j)}$. 

\begin{definition} \label{def:exactpromisedthermalstate}  Let $M$ be a rounding promise. Let the \textbf{exact promised partition function} $ \hat{\mathcal{Z}}^{(M)}_\beta$ and \textbf{exact promised thermal state} $ \hat{\rho}^{(M)}_\beta$ be defined similarly as their non-exact versions, just leaving $\lambda_i$ intact rather than swapping them out with the midpoints $m_x$ of the intervals containing them.
\begin{align}
    \hat{\mathcal{Z}}^{(M)}_\beta &\coloneqq \sum_{i}  1_{\lambda_i \in M} \exp(-\beta \lambda_i), \quad\text{and} \\
    \hat{\rho}^{(M)}_\beta&\coloneqq  \sum_i  1_{\lambda_i \in M} \Pi_i \exp(-\beta \lambda_i) / \hat{\mathcal{Z}}^{(M)}_\beta.
\end{align}
\end{definition}

These exact promised thermal states are still restricted to the promised subspace, but unlike the regular promised thermal states ${\rho_\beta}^{(M_j)}$ their energies are unrounded and left intact. Thus, the difference between ${\hat{\rho}_\beta}^{(M_j)}$ and ${\rho_\beta}^{(M_j)}$ captures the error from energy estimation only. We will show that $\norm{ {\hat \rho_\beta}^{(M_j)} - {\rho_\beta}^{(M_j)}}_1 \leq \sqrt{\beta 2^{-n}}$. Then, we consider the ensemble $\hat \rho^{*M} $ of the states ${\hat \rho_\beta}^{(M_j)} $. It follows that  $\norm{ \hat \rho^{*M} -  \rho^{*M}}_1 \leq \sqrt{\beta 2^{-n}}$.

Second, we bound the error stemming from the elimination of certain energies in the rounding promises. In the previous subsection, we established that the probability that any particular energy $\lambda$ is missing from $M_j$ is at most $1/2^r$. Through careful consideration of the promised partition functions ${\mathcal{Z}_\beta}^{(M_j)}$, exact promised partition functions ${\hat{\mathcal{Z}}}^{(M_j)}_\beta$ and the ideal partition function $\mathcal{Z}_\beta$, we use this property to show that $\norm{ \hat \rho^{*M} -  \rho_\beta } \leq 2\cdot 2^{-r}$. 

Combining these two facts yields the main theorem of this subsection:
\begin{theorem}[\textbf{Accuracy of the final ensemble}] \label{thm:accuracyofthefinalensemble} Consider the following ensembles of density matrices:
\begin{align}
    \rho^{*M} \coloneqq \frac{1}{2^r} \sum_{j=0}^{2^r-1} \rho_{\beta}^{(M_j)} \quad
    \text{for } M\in\{L,R\}.
\end{align}
Then, we have $\norm{\rho^{*M} - \rho_\beta}_1 \leq \sqrt{\beta  2^{-n}} + 2\cdot 2^{-r}$,
\end{theorem}
\begin{proof} Combining \cref{lemma:roundingerrorsfromenergyestimation,lemma:theaverageofexactpromisedthermalstatesapproximatesthethermalstate} with the triangle inequality, we obtain $\norm{\rho^{*M} - \rho_\beta}_1 \leq \sqrt{\beta  2^{-n}} + 2\cdot 2^{-r}$ for both $M \in \{L,R\}$. 
\end{proof}

We start with the analysis of the rounding errors via the exact promised thermal states. To prove that these are close together, we leverage a fact from \cite{0905.2199}:
\begin{lemma}[\textbf{Comparing thermal states via their Hamiltonians}~{\cite[Appendix C]{0905.2199}}]
\label{lemma:comparingthermalstatesviatheirhamiltonians}
Say $H$ and $H'$ are two Hamiltonians. Let $\rho_\beta \coloneqq e^{-\beta H}/\Tr(e^{-\beta H})$ and $\rho_\beta' \coloneqq e^{-\beta H'}/\Tr(e^{-\beta H'})$ be their Gibbs states, respectively. Then, if $F(\sigma,\sigma') = \norm{\sqrt{\sigma}\sqrt{\sigma'}}_1$ is the fidelity:
\begin{align}
    F(\rho,\rho') \geq \exp(-\beta |H - H'|  ).
\end{align}
\end{lemma}

This tool is convenient because a naive analysis of the trace distance yields an contribution from every eigenvalue, which accumulates into a total error proportional to the dimension. By considering the spectral distance in the Hamiltonian instead, we avoid the exponential blowup in the error. However, this trick is a bit tricky in our case since, when viewed as operators on the full Hilbert space, ${\rho_\beta}^{(M_j)}$ is not full rank. But states of the form $e^{-\beta H}/\text{Tr}(e^{-\beta H})$ are always full rank whenever $H$ is.

To remedy this issue, we recall an observation made in \cref{remark:wellroundedhamiltonian}, where we noticed that if we view ${\rho_\beta}^{(M_j)}$ as a density matrix over $\mathcal{P}^{(M_j)}$ only, then it can be written as $e^{-\beta  H^{(M_j)}}/\text{Tr}(e^{-\beta \bar H^{(M_j)}})$ where $ H^{(M_j)}$ is the promised Hamiltonian. We will leverage this viewpoint where we restrict to $\mathcal{P}^{(M_j)}$ for \cref{def:exactpromisedhamiltonian} and \cref{lemma:comparingthermalstatesviatheirhamiltonians}.

\begin{definition} \label{def:exactpromisedhamiltonian}  Let $M$ be a rounding promise. Let the \textbf{exact promised Hamiltonian} $\hat H^{(M)}$ be the part of $H$ with support on the promised subspace $\mathcal{P}^{(M)}$. But unlike the non-exact promised Hamiltonian $H^{(M)}$, the eigenvalues are unchanged. It is given by:
\begin{align}
    \hat H^{(M)} &\coloneqq \sum_{i}  \lambda_i \Pi_i.
\end{align}
Here, we view this operator as $\hat H^{(M)} \in \mathrm{L}(\mathcal{P}^{(M)})$.
\end{definition}

\begin{lemma}[\textbf{Rounding errors from energy estimation}]
\label{lemma:roundingerrorsfromenergyestimation} For $M \in \{L,R\}$ consider the following ensemble over exact promised thermal states:
\begin{align}
        \hat{\rho}^{*M} \coloneqq 2^{-r} \sum_j \hat{\rho}^{(M_j)}_\beta.
\end{align}
This density matrix satisfies $\norm{ \hat{\rho}^{*M} - \rho^{*M} }_1 \leq \sqrt{\beta 2^{-n}}$.
\end{lemma}
\begin{proof} Recall from \cref{lemma:propertiesofthecoarsegrainedroundingpromises} that each of the intervals $[a_x,b_x]$ of the $M_j$ satisfies $b_x - a_x \leq 2^{-n}$. Thus, if $\lambda_i \in [a_x,b_x]$, we have $| \lambda_i - m_x | \leq 2^{-n-1}$ where $m_x$ is the midpoint of $[a_x,b_x]$. Consequently, we have $\norm{\hat H^{(M_j)} - H^{(M_j)}} \leq 2^{-n-1}$.

Following \cref{lemma:comparingthermalstatesviatheirhamiltonians}, we have $F(\hat{\rho}^{(M_j)}_\beta,\rho_\beta^{(M_j)}) \geq e^{-\beta 2^{-n}} $, so, exploiting the fact that $\norm{\sigma - \sigma'}_1 \leq \sqrt{1-F(\sigma,\sigma')^2}$, we have:
\begin{align}
    \norm{\hat{\rho}^{(M_j)}_\beta - \rho_\beta^{(M_j)}}_1 \leq \sqrt{1 - F(\hat{\rho}^{(M_j)}_\beta,\rho_\beta^{(M_j)})^2} \leq \sqrt{1 - e^{-2\beta 2^{-n-1} }} \leq \sqrt{\beta 2^{-n} }.
\end{align}
Finally, we average this error over the two ensembles:
\begin{align}
    \norm{ \hat{\rho}^{*M} - \rho^{*M} }_1 = \norm{ 2^{-r}\sum_j \left(\hat{\rho}^{(M_j)}_\beta - \rho^{(M_j)}_\beta \right)  }_1  \leq 2^{-r}\sum_j\norm{\hat{\rho}^{(M_j)}_\beta - \rho^{(M_j)}_\beta   }_1 \leq \sqrt{\beta 2^{-n}}.
\end{align}
\end{proof}

We move on to the error stemming from the elimination of energies in the promises $M_j$.

\begin{lemma}[\textbf{The average of exact promised thermal states approximates the thermal state}]
\label{lemma:theaverageofexactpromisedthermalstatesapproximatesthethermalstate}
For $M \in \{L,R\}$, we have $\norm{\hat\rho^{*M} - \rho_\beta}_1 \leq   2\cdot 2^{-r}$. 
\end{lemma}
\begin{proof}  Letting $d_i \coloneqq \text{Tr}(\Pi_i)$ be the occupation numbers, we have:

    \begin{align}
        \norm{  \hat\rho^{*(M)} - \rho_\beta }_1 &= \sum_i  d_i \left| 2^{-r}  \sum_{j} 1_{\lambda_i \in M_j}  \exp\left( - \beta \lambda_i \right)/\hat{\mathcal{Z}}^{(M_j)}_\beta  - \exp( -\beta \lambda_i) / \mathcal{Z}_\beta \right|\\
&\leq\sum_i  d_i \left| 2^{-r}  \sum_{\substack{j \\ \lambda_i \not\in M_j}}\left( 0  - \exp( -\beta \lambda_i) / \mathcal{Z}_\beta \right) \right|\\
        &+ \sum_i  d_i \left| 2^{-r}  \sum_{\substack{j \\ \lambda_i \in M_j}}  \left( \exp\left( - \beta \lambda_i \right)/\hat{\mathcal{Z}}^{(M_j)}_\beta  - \exp( -\beta \lambda_i) / \mathcal{Z}_\beta  \right)\right|.
    \end{align}

    Here, for each $\lambda_i$, we split the rounding promises $M_j$ into two categories: those not containing $\lambda_i$, for which $\hat\rho^{(M)}_\beta$ has eigenvalue $0$, and those containing $\lambda_i$, for which $\hat{\rho}_\beta^{(M)}$ has eigenvalue $\exp\left( - \beta \lambda_i \right)/\hat{\mathcal{Z}_\beta^{(M_j)}}$. After applying the triangle inequality to these two groups, these correspond to the first and second terms in the above sum, respectively.  We will show that both of these terms are bounded by $2^{-r}$, so the overall bound is $2\cdot2^{-r}$ as desired.

The first term is bounded by:
\begin{align}
    \sum_i  d_i \left| 2^{-r}  \sum_{\substack{j\\ \lambda_i \not\in M_j}}\left( 0  - \exp( -\beta \lambda_j) / \mathcal{Z}_\beta \right) \right| \leq 2^{-r} \cdot \frac{1}{\mathcal{Z}_\beta} \sum_i d_i  \exp( -\beta \lambda_i) =2^{-r} \cdot \frac{\mathcal{Z}_\beta}{\mathcal{Z}_\beta} = 2^{-r}.
\end{align}

    So all that is left is the second term. Observe that $\hat{\mathcal{Z}}^{(M_j)}_\beta  \leq \mathcal{Z}_\beta   $ which implies that $1 /\hat{\mathcal{Z}}^{(M_j)}_\beta  - 1 / \mathcal{Z}_\beta \geq 0$. This allows us to remove the absolute value:
    \begin{align}
        & \sum_i  d_i \left| 2^{-r}  \sum_{\substack{j \\ \lambda_i \in M_j}}  \left( \exp\left( - \beta \lambda_i \right)/\hat{\mathcal{Z}}^{(M_j)}_\beta  - \exp( -\beta \lambda_i) / \mathcal{Z}_\beta  \right) \right|\\
        &= \sum_i  d_i \exp\left( - \beta \lambda_i \right) \cdot \left|  2^{-r}  \sum_{j}  1_{\lambda_i \in M_j}  \left( 1/\hat{\mathcal{Z}}^{(M_j)}_\beta  - 1 / \mathcal{Z}_\beta  \right)\right|\\
        &= \sum_i  d_i   \exp\left( - \beta \lambda_i \right) \cdot 2^{-r}\sum_{j} 1_{\lambda_i \in M_j}  \left(1 /\hat{\mathcal{Z}}^{(M_j)}_\beta  - 1 / \mathcal{Z}_\beta  \right).
\end{align}

    Now we can bound the two terms corresponding to the difference $1 /\hat{\mathcal{Z}}^{(M_j)}_\beta  - 1 / \mathcal{Z}_\beta  $ individually. The term with $1 / \hat{\mathcal{Z}}^{(M_j)}_\beta$ is is just 1:
\begin{align}
&\sum_i  d_i   \exp\left( - \beta \lambda_i \right) \cdot 2^{-r}\sum_{j} 1_{\lambda_i \in M_j}  \left(1 /\mathcal{Z}^{(M_j)}_\beta  \right)\\
&= 2^{-r}\sum_{j} \left[ \sum_i  d_i   \exp\left( - \beta \lambda_i \right)  1_{\lambda_i \in M_j} \right] \left(1 /\mathcal{Z}^{(M_j)}_\beta  \right)\\
&= 2^{-r}\sum_{j} \mathcal{Z}^{(M_j)}_\beta  \left(1 /\mathcal{Z}^{(M_j)}_\beta  \right) = 1.
\end{align}

    The term with $-1 / \mathcal{Z}_\beta$ can be bounded by leveraging that for any $\lambda_i$, at most one $M^{(j)}$ doesn't contain it, as shown in \cref{lemma:propertiesofthecoarsegrainedroundingpromises}. That is, $2^{-r} \sum_j 1_{\lambda_i \in M_j} \geq 1 - 2^{-r}$.
\begin{align}
&- \sum_i  d_i   \exp\left( - \beta \lambda_i \right) \cdot 2^{-r}\sum_{j} 1_{\lambda_i \in M_j}  \left(1 /\mathcal{Z}_\beta  \right)\\
&\leq - \sum_i  d_i   \exp\left( - \beta \lambda_i \right) \cdot (1-2^{-r})  \left(1 /\mathcal{Z}_\beta  \right)\\
&= - \mathcal{Z}_\beta \cdot (1-2^{-r})  \left(1 /\mathcal{Z}_\beta  \right) = - (1-2^{-r}).
\end{align}
So a bound on the second term is $1 - (1-2^{-r}) = 2^{-r}$. So, each of the two terms is at most $2^{-r}$, so the total error is at most $2\cdot 2^{-r}$.
\end{proof}

\section{Lindblad dynamics on the promised subspace\label{sec:lindbladdynamicsonthepromisedsubspace}}

In the previous section we required the following capability: to take a state $\sigma^{(M)}$ that is supported on a rounding promise $M$, and to transform it to the promised thermal state ${\rho_\beta}^{(M)}$. Once we have this capability, we can prepare good approximations to the true thermal state $\rho_\beta$.

We attain this capability by running a promised Davies generator $\mathcal{L}^{(M)}$ on the system, which we defined in \cref{def:promiseddaviesgenerator}. We constructed $\mathcal{L}^{(M)}$ by starting with a regular Davies generator $\mathcal{L}$ and truncating it  so that its dynamics are restricted to $\mathcal{P}$. The basic idea is to consider the coupling operators $S_\alpha$ of the original Davies generator, and to conjugate them with an attenuation operator $A^{(M)}$ to obtain:
\begin{align}
    S_\alpha^{(M)} \coloneqq A^{(M)} S_\alpha A^{(M)}
\end{align}
Then, $\mathcal{L}^{(M)}$ is defined the exact same way as $\mathcal{L}$, just with $ {S_\alpha}^{(M)}$ instead of $S_\alpha$. The other difference is that $\mathcal{L}$ is defined with respect to the ideal Hamiltonian $H = \sum_i \lambda_i \Pi_i$, and $\mathcal{L}^{(M)}$ with respect to the promised Hamiltonian $H^{(M)} = \sum_x m_x \Pi_x$.

In this section we detail the construction of the attenuation operator $A^{(M)}$ and how that the resulting promised Davies generator has some desirable properties.  Ideally, we would like to have $A^{(M)} = P^{(M)}$, but due to limitations of singular value transformation, we cannot achieve this. This results in two challenges.

 The first challenge is \emph{leakage}: we cannot perfectly eliminate the transfer of probability mass from $\mathcal{P}$ to $\mathcal{P}^\perp$; we can only suppress it by a factor $\delta_\mathrm{leak}$. Fortunately, this error is both exponentially suppressible, and its ramifications on the final output error can be quantified formally. We give a construction of an approximate attenuation operator $\tilde A^{(M)}$ that has support at most $\delta_\mathrm{leak}$ on $\mathcal{P}^\perp$. The operator $A^{(M)}$ itself has no support on $\mathcal{P}^\perp$ at all, ensuring that $\mathcal{L}^{(M)}$ does not leak. The leakage error can be dealt with by studying its approximate implementation $\tilde{\mathcal{L}}^{(M)}$ which we construct in \cref{sec:implementinglindbladdynamics}.

The second challenge, which is more significant, is \emph{attenuation}. Since $A^{(M)}$ is implemented via singular value transformation, it takes the form $A^{(M)} = \sum_{i} a^{(M)}(\lambda_i) \Pi_i$ for some continuous function $a^{(M)}:[0,1]\to [0,1]$. We have already achieved that $a^{(M)}(\lambda) = 0$ for $\lambda \not\in M$. Ideally, we would like to also have $a^{(M)}(\lambda) = 1$ for $\lambda \in M$, but this means that $a^{(M)}$ has a discontinuity at the boundaries of $M$. The function $a^{(M)}$ must be continuous, so we need to smoothly interpolate from 0 to 1 near the edges of the intervals of $M$. Formally, we can define a truncated rounding promise $M^T$ whose intervals are a little bit narrower than those of $M$, and then ensure $a^{(M)}(\lambda) = 1$ for $\lambda \in M^T$.  However, this choice means that the operator also attenuates certain eigenvectors of $H$ whose eigenvalues are close to the edges of the intervals of $M$. This attenuation might slow the rate of convergence of $\mathcal{L}^{(M)}$. 

Having highlighted the trade-offs of the construction, we formally define the attenuation operator $A^{(M)}$ and its approximation $\tilde A^{(M)}$.

\begin{lemma}[\textbf{Attenuation operators}]
\label{lemma:attenuationoperators}
Consider any rounding promise $M$ with minimum gap width $\kappa$ whose intervals are also at least $\kappa$ wide. Then, for any leakage error $\delta_\mathrm{leak}>0$, and attenuation factor $\gamma$ there exists an \textbf{attenuation operator} $A^{(M)}$ and an \textbf{approximate attenuation operator} $\tilde A^{(M)}$. They both commute with the Hamiltonian, and take the form:
\begin{align}
    A^{(M)} = \sum_{i} a^{(M)}(\lambda_i) \Pi_i \hspace{1cm}\text{and} \hspace{1cm} \tilde A^{(M)} = \sum_{i} \tilde a^{(M)}(\lambda_i) \Pi_i
\end{align}
for some functions $a^{(M)},\tilde a^{(M)}:[0,1]\to [0,1]$. Let $M^T$ be a $\kappa\gamma$-truncated rounding promise where if $[a_x, b_x] \in M$, then $[a_x+\kappa\gamma, b_x - \kappa\gamma  ] \in M^{T}$. Then, $a^{(M)},\tilde a^{(M)}:[0,1]\to [0,1]$ satisfy:
\begin{align}
    \text{if } \lambda\not\in M\text{ then }& a^{(M)}(\lambda) = 0 \text{ and } \tilde a^{(M)}(\lambda) \leq \delta_\mathrm{leak},\\
    \text{if } \lambda\in M^T \text{ then }& a^{(M)}(\lambda) = 1 \text{ and } \tilde a^{(M)}(\lambda) \geq 1- \delta_\mathrm{leak},\\
    \text{otherwise }& a^{(M)}(\lambda) =  \tilde a^{(M)}(\lambda).
\end{align}
Consequently, we have $\norm*{A^{(M)} - \tilde A^{(M)}} \leq \delta_\mathrm{leak}$, and $\norm*{ A^{(M)} (I - P^{(M)}) } = 0$ as required by \cref{def:promiseddaviesgenerator}. Furthermore, there exists a block encoding of $\tilde A^{(M)}$ with circuit complexity $O\left( \kappa^{-1}\gamma^{-1} \log( s^{(M)}\delta_\mathrm{leak}^{-1})\right)$.
\end{lemma}
The proof is in \cref{app:polynomialconstruction}. We will show how to deal with leakage errors in \cref{sec:implementinglindbladdynamics}. This section is dedicated to assessing the impact of truncation and attenuation on the dynamics of $\mathcal{L}^{(M)}$. Our goal is to show that $\mathcal{L}^{(M)}$ is mixing (recall \cref{def:mixing}) and that the mixing time $t_\mathrm{mix}$ is not too much slower than that of $\mathcal{L}$.

On a positive note, we emphasize that truncation to $\mathcal{P}$ as well as attenuation cannot cause $\mathcal{L}^{(M)}$ to fail to converge onto the promised thermal state unless the coupling operators $S_\alpha$ are extremely contrived. For example, one way $\mathcal{L}$ could be mixing while $\mathcal{L}^{(M)}$ could fail to be so is if the $S_\alpha$ encode a particular random walk over the energy eigenspaces. Specifically, if the random walk is over a bipartite graph where one half of the eigenspaces are in $\mathcal{P}$ and the other half are in $\mathcal{P}^\perp$, then a Davies generator based on $S_\alpha$ may be mixing, while $S^{(M)}_\alpha$ will not be. But construction of such a coupling operator requires careful knowledge over the energy eigenspaces. The chances of accidentally selecting such an $S_\alpha$ in practice are slim.

The main problem is that the projection of $S_\alpha$ to ${S_\alpha}^{(M)}$ will slow the rate of mixing, and that attenuation only makes this problem worse. We assess the severity of this problem by performing numerical simulations on a particular physical system: the transverse field Ising model. Let $(n_1,n_2)$ be the width and height of 2d-grid of qubits, and let $v$ be the strength of the transverse field. Let $\vec i$ be indices in an $n_1 \times n_2$ grid, let $X_{\vec i}$ and $Z_{\vec i}$ denote Pauli $X$ and $Z$ on the qubit at grid position $\vec i$, and let $\vec i\sim \vec j$ denote that the two grid positions are adjacent. Then, the Hamiltonian we consider in this section is:

\begin{align}
    H = \sum_{\vec i \sim \vec j} Z_{\vec i} Z_{\vec j} +  v \sum_{\vec i} X_{\vec i}. \label{eqn:isingmodel}
\end{align}

We study the dynamics of Davies generators defined with respect to this Hamiltonian, as well as the coupling operators $X_{\vec i}$ and $Z_{\vec i}$ for all grid positions $\vec i$. To assess the mixing time $t_\mathrm{mix}$ we compute the spectral gap $\Delta$ of the Lindbladian.  

\begin{figure}
    \centering
    \includegraphics[width=\textwidth]{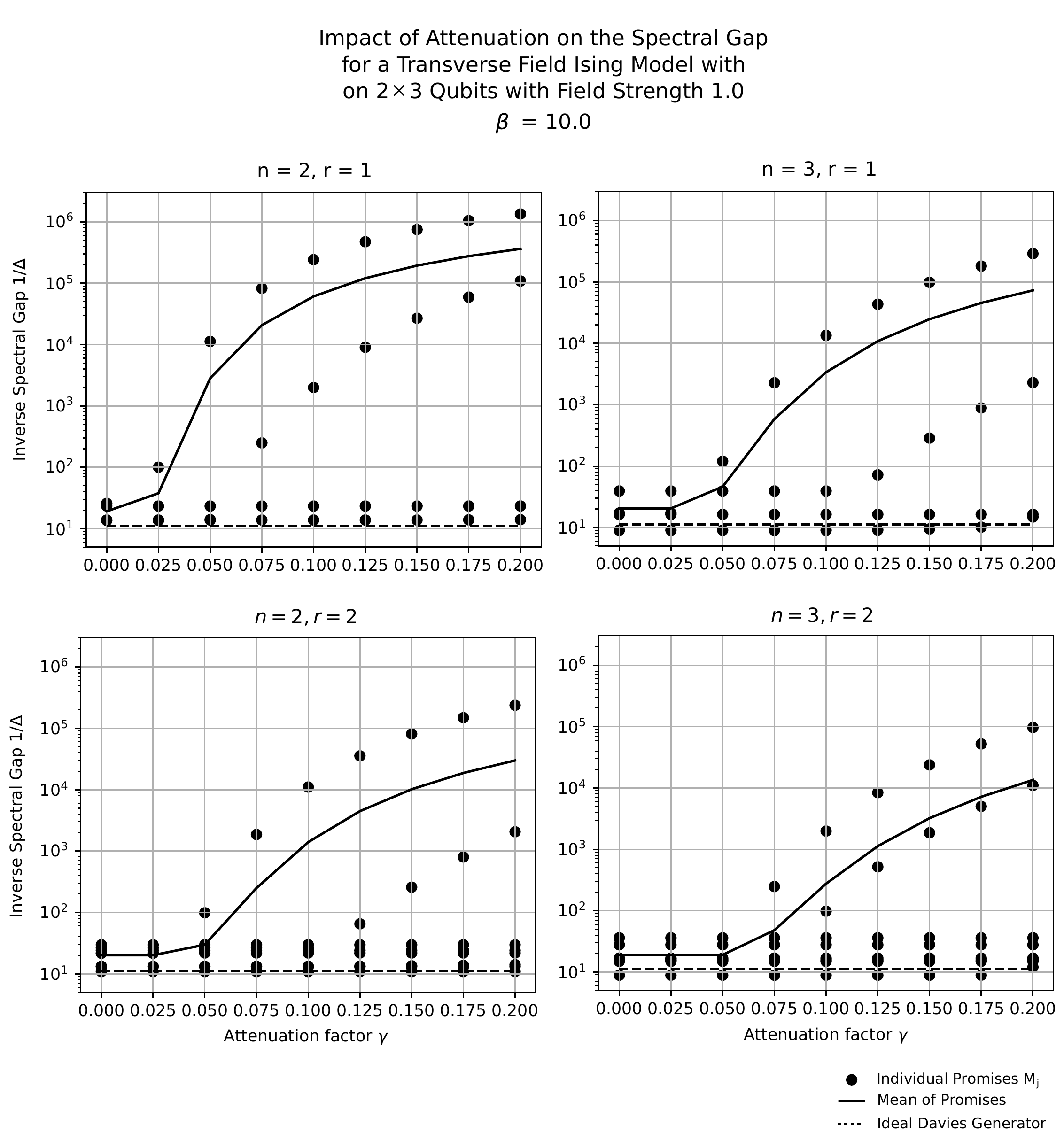}
    \caption{Analysis of the spectral gap of the Davies generators of the transverse field Ising model from \cref{eqn:isingmodel}. We set $\beta = 10.0$ throughout. We see that the convergence rates of the promised Davies generators are a little slower than the ideal Davies generator as $\gamma \to 0$. While the convergence rate slows down significantly with large $\gamma$, selecting reasonably small values of $\gamma$ achieves a convergence time just as fast as the unphysical case of $\gamma = 0$. There is also significant variation among the promises $M_j$. The software that produced these data is available at: \url{https://github.com/qiskit-community/promised-davies-generator}.}
    \label{fig:attenuationanalysis}
\end{figure}

In \cref{fig:attenuationanalysis} we plot spectral gap as a function of the attenuation factor $\gamma$ for various parameters. In all our experiments, we verified that the ideal Davies generator $\mathcal{L}$ had the ideal thermal state $\rho_\beta$ as its unique fixed point, and that the promised Davies generators $\mathcal{L}^{(M_j)}$ had the promised thermal state $\rho^{(M_j)}_\beta$ as their unique fixed point. We analyze their exact versions with $A^{(M)}$ rather than $\tilde A^{(M)}$ and $P_x$ rather than $\tilde P_x$.

We find that in the ideal case of $\gamma = 0$, the ideal Davies generator and the promised Davies generators have about the same convergence rate - the promised Davies generators are a little slower. So while there is some slowdown in mixing time from projection, the slowdown is largely due to attenuation. We find that setting reasonably small values of $\gamma$ achieves convergence rates similar to the $\gamma = 0$ case. Furthermore, this slowdown only appears to be significant for certain rounding promises $M_j$.  One interpretation of our technique in this paper is that we trade mixing time for a provable accuracy guarantee on the final thermal state. So, some amount of slowdown is acceptable in exchange for the rigor. 
\newpage

\section{Implementing Lindblad dynamics\label{sec:implementinglindbladdynamics}}
The goal of this section is to show how to simulate the evolution of the promised Davies generator for time $t$ given a block encoding of its jump operators. This block encoding can be implemented by energy estimation. This construction also involves the coupling operator projected onto the promised subspace. Both energy estimation and the projection are not perfect, and they lead to errors in the block-encoded jump operators. We begin with a lemma that shows that a small perturbation of the jump operator will not change the Lindbladian by much. Then we give the detailed construction and prove the main theorem.

For simplicity, in this section we assume that there is only one coupling operator so that we can drop the index $\alpha$ in for the jump operators $L_{\nu, \alpha}$ and the coupling operators $S_{\alpha}$, and simply write $L_{\nu}$ for the jump operators of a promised Davies generator (\cref{def:promiseddaviesgenerator}), and write $S$ for the coupling operator. Our analysis can be easily generalized to the case of more coupling operators: the number of jump operators just grows by the factor of the number of the coupling operators.

Furthermore, it is reasonable to suppose that $\norm{S} \leq 1$. First, this is generally the case in practice: the coupling operator $S$ is usually selected to be a unitary transformation that scrambles the energy eigenbasis. Even if $S$ were a block encoded operator, it would require a unitary implementation, so the only way to achieve $\norm{S} \geq 1$ would be via some virtual scale factor.  Second, in the analysis of Hamiltonian simulation $e^{i H t}$, it is usually assumed that $\norm{H} \leq 1$, since we can always absorb a rescaling of the Hamiltonian into the time $t$. The exact same argument applies for the simulation of Davies generators.

To begin with, we show that the distance between Lindbladians in terms of the diamond norm can be bounded by the distances of their jump operators in terms of the spectral norm.
\begin{lemma}[\textbf{Approximating Lindblad evolution via jump operators}] 
\label{lemma:jumpapprox} 
  Say $L_\nu \in \mathrm{L}(\mathcal{H}^{(M)})$ are jump operators defining a Lindbladian $\mathcal{L}$, and say there are $m$ many of them. Say some $\tilde L_\nu \in \mathrm{L}(\mathcal{H})$ similarly define a Lindbladian $\tilde{\mathcal{L}}$, and say these satisfy:
\begin{align}
  \label{eq:ll}
    \norm*{\tilde L_\nu - L_\nu } \leq \delta_L.
\end{align}

Say we also have $\norm{L_\nu} \leq 1$. Then, for any $t > 0$:

\begin{align}
    \norm*{\tilde{\mathcal{L}} - \mathcal{L}}_\diamond \leq  4m\delta_L  + 2\delta_L^2
\end{align}
\end{lemma}
\begin{proof}
  First of all, by \cref{eq:ll} and the triangle inequality, we have
  \begin{align}
    \norm*{\tilde{L}_{\nu}} \leq \norm{L_{\nu}} + \delta_L.
  \end{align}
  Now, recall that
  \begin{align}
    \mathcal{L}(\rho) &= \sum_{\nu} L_{\nu} \rho  L_{\nu}^{\dag} - \frac{1}{2}\left(L_{\nu}^{\dag}L_{\nu} \rho +  \rho L_{\nu}^{\dag}L_{\nu}\right), \quad\text{ and } \\
    \tilde{\mathcal{L}}(\rho) &= \sum_{\nu} \tilde{L}_{\nu}  \rho  \tilde L_{\nu} - \frac{1}{2}\left(\tilde L_{\nu}^{\dag}\tilde L_{\nu} \rho +  \rho \tilde L_{\nu}^{\dag}\tilde L_{\nu}\right).
  \end{align}
  Let $\rho$ be the matrix that achieves the induced trace norm of $\norm*{\mathcal{L} - \mathcal{\tilde{L}}}_1$, i.e.,
  \begin{align}
    \norm*{\mathcal{L} - \mathcal{\tilde{L}}}_1 = \norm*{\mathcal{L}(\rho) - \mathcal{\tilde{L}}(\rho)}_1.
  \end{align}
  We have
  \begin{align}
    \norm*{\mathcal{L} - \tilde{\mathcal{L}}}_1 &\leq \sum_{\nu} \norm{L_{\nu}\rho L_{\nu}^{\dag} - \tilde{L}_{\nu}\rho\tilde{L}_{\nu}^{\dag}}_1 + \frac{1}{2}\norm{\tilde{L}_{\nu}^{\dag}\tilde{L}_{\nu}\rho - L_{\nu}^{\dag}L_{\nu}\rho}_1 + \frac{1}{2}\norm{\rho\tilde{L}_{\nu}^{\dag}\tilde{L}_{\nu} - \rho L_{\nu}^{\dag}L_{\nu}}_1 \\
   &\leq \sum_{\nu} \norm{L_{\nu}}\delta_L + \norm*{\tilde{L}_{\nu}}\delta_L + \norm*{\tilde{L}_{\nu}\tilde{L}_{\nu} - L_{\nu}^{\dag}L_{\nu}} \\
   &\leq \sum_{\nu} 2(\norm{L_{\nu}}\delta_L + \norm*{\tilde{L}_{\nu}}\delta_L) \\
   &\leq 4m\delta_L + 2 \delta^2_L,
  \end{align}
  where we have used the fact that $\norm{ABC}_1 \leq \norm{A}\norm{B}_1\norm{C}$ for matrices $A, B$ and $C$.

  To extend this bound to the diamond norm, we consider the identity map $\mathcal{I}: \mathrm{L}(\mathcal{H}^{(M)}) \rightarrow \mathrm{L}(\mathcal{H}^{(M)})$. We have
  \begin{align}
    \mathcal{L}\otimes\mathcal{I}(\rho) &= \sum_{\nu} (L_{\nu}\otimes I) \rho  (L_{\nu}^{\dag}\otimes I) - \frac{1}{2}\left((L_{\nu}^{\dag}L_{\nu}\otimes I) \rho +  \rho (L_{\nu}^{\dag}L_{\nu}\otimes I)\right), \quad\text{ and } \\
    \tilde{\mathcal{L}}\otimes\mathcal{I}(\rho) &= \sum_{\nu} (\tilde{L}_{\nu}\otimes I)  \rho  (\tilde L_{\nu}^{\dag}\otimes I) - \frac{1}{2}\left((\tilde L_{\nu}^{\dag}\tilde L_{\nu}\otimes I) \rho +  \rho (\tilde L_{\nu}^{\dag}\tilde L_{\nu}\otimes I)\right).
  \end{align}
  Then it is easy to obtain that
  \begin{align}
    \norm*{\mathcal{L} - \mathcal{\tilde{L}}}_{\diamond} = \norm*{\mathcal{L}\otimes\mathcal{I} - \tilde{\mathcal{L}}\otimes\mathcal{I}}_1 \leq  4m\delta_L  + 2 \delta^2_L,
  \end{align}
  using the fact that $\norm{A\otimes I} = \norm{A}$ for any matrix $A$.
\end{proof}

The Lindbladian we try to simulate is the promised Davies generator $\mathcal{L}^{(M)}$, which involves the promised coupling operator $S^{(M)} = A^{(M)}SA^{(M)}$. However, in our construction, we implement an approximate attenuation operator $\tilde S^{(M)} = \tilde A^{(M)}S\tilde A^{(M)}$ (see \cref{sec:lindbladdynamicsonthepromisedsubspace}). Based on the distance between $A^{(M)}$ and $\tilde A^{(M)}$ given by \cref{lemma:attenuationoperators}, we prove a bound on the distance between $S^{(M)}$ and $\tilde{S}^{(M)}$ as follows.
\begin{lemma}
\label{lemma:distance-s-tildes}
Given a rounding promise $M$, let $S^{(M)}$ and $\tilde S^{(M)}$ be as defined in \cref{sec:lindbladdynamicsonthepromisedsubspace}. For any leakage error $\delta_{\mathrm{leak}} > 0$, it holds that
\begin{align}
    \norm{S^{(M)} - \tilde S^{(M)}} \leq 2\delta_{\mathrm{leak}}.
\end{align}
\end{lemma}
\begin{proof}
Recall that $S^{(M)} = A^{(M)}SA^{(M)}$, and $\tilde S^{(M)} = \tilde A^{(M)}S\tilde A^{((M)}$. By \cref{lemma:attenuationoperators}, we have
\begin{align}
    \norm{A^{(M)} - \tilde A^{(M)}} \leq \delta_{\mathrm{leak}}.
\end{align}

To prove the desired bound, we leverage  $\norm{A^{(M)}}\leq 1$ and  $\norm{\tilde{A}^{(M)}}\leq 1$ to show:
\begin{align}
  \norm{S^{(M)} - \tilde S^{(M)}} &= \norm{A^{(M)}SA^{(M)} - \tilde A^{(M)}S\tilde A^{(M)}} \\
                                  &= \norm{A^{(M)}SA^{(M)} - A^{(M)}S\tilde A^{(M)} + A^{(M)}S\tilde A^{(M)} - \tilde A^{(M)}S\tilde A^{(M)}} \\
                                  &\leq \delta_{\mathrm{leak}} + \delta_{\mathrm{leak}} = 2\delta_{\mathrm{leak}}.
\end{align}
Note that we have assumed that $\norm{S} \leq 1$ as in the beginning of this section.
\end{proof}

Now, we have all the tools to prove the main theorem and present our quantum algorithm for simulating the Davies generator.

\begin{theorem}[\textbf{Implementation of the Davies generator given a rounding promise}] For some rounding promise $M$, say $\mathcal{L}^{(M)}$ is an promised Davies generator as defined in \cref{def:promiseddaviesgenerator}. Then, for any $\delta_\mathcal{L},t>0$, there exists a quantum algorithm implementing a channel $\delta_\mathcal{L}$-close in diamond norm to $e^{t\mathcal{L}^{(M)}}$. If $\kappa$ is the minimum gap of $M$, $M$ has $s^{(M)}$ many intervals, and $\gamma$ is the desired attenuation factor for the $\tilde S^{(M)}$, then this algorithm uses
\begin{align}
  \label{eq:gc-L}
  O\left((s^{(M)})^2  t\,\frac{\log (s^{(M)} t/\delta_{\mathcal{L}})}{\log\log (s^{(M)}t/\delta_{\mathcal{L}})}\right)
\end{align}
queries to the block encoding of $S$,
\begin{align}
  O\left( t \cdot \kappa^{-1} \cdot \gamma^{-1}  \cdot  (s^{(M)})^2\log(s^{(M)}) \cdot \frac{\log^2(s^{(M)}t/\delta_{\mathcal{L}})}{\log\log (s^{(M)}t/\delta_{\mathcal{L}})}\right)
\end{align}
queries to the block encoding of $H$, and
\begin{align}
  \label{eq:total-gc}
  O\left((s^{(M)})^4t\,\left(\frac{\log (s^{(M)} t/\delta_{\mathcal{L}})}{\log\log (s^{(M)} t/\delta_{\mathcal{L}})}\right)^2\right)
\end{align}
additional 1- and 2-qubit gates.

\end{theorem}

\begin{proof} We aim to invoke \cref{prop:lindbladsimulation}, so we need an oracle $\mathcal{O}_{\tilde{\mathcal{L}}}$. This lets us approximately apply the map $e^{t\tilde{\mathcal{L}}}$. We then show that the $\tilde L_\nu$ defining $\tilde{\mathcal{L}}$ are close to the $L_\nu$ defining $\mathcal{L}^{(M)}$, allowing us to invoke \cref{lemma:jumpapprox}.

Recall the approximate energy estimation isometry $\tilde E^{(M)}$ from \cref{prop:energyestimationgivenaroundingpromise}, which extracts estimates according to $\tilde P_x$ satisfying $|\tilde P_x P^{(M)} - P_x| \leq \delta_\mathrm{est}$ via isometries $G_{\tilde P_x}$ satisfying $G^\dagger_{\tilde P_x}G_{\tilde P_x} = \tilde P_x$. We define the block-encoded operator:
    \begin{align}
        V^{(M)}_{\pm}  \coloneqq \hspace{5mm}   \begin{array}{c}\Qcircuit @R=2mm @C=2mm{  &  \qw & \gate{\pm}  & \qw & \qw & \\
              &   \multigate{1}{\tilde E^{(M)}} & \ctrl{-1} & \multigate{1}{\tilde  E^{(M)\dagger}} & \rstick{\hspace{-42.8mm}\textcolor{white}{\rule{3mm}{1mm}}}   \\
               &   \ghost{\tilde E^{(M)}} & \qw & \ghost{\tilde E^{(M)\dagger}}  & \qw &
        }\end{array}
    \end{align}
    where the controlled $\pm$ refers to adding/subtracting ${m_x}^{(M)}$ to/from the value of the target register (recall \cref{remark:wellroundedhamiltonian}, the control register has the value of $x$), and the kets on the right side of the circuit denote postselection onto that state. Next, obtain a block encoding of:
            \begin{align}
            G \coloneqq \sum_\nu \sqrt{G(\nu)} \otimes \ketbra{\nu}{\nu}.
            \end{align}
            Finally, let $\mathcal{O}_{\tilde{\mathcal{L}}}$ be given by:
             \begin{align}
                \mathcal{O}_{\tilde{\mathcal{L}}}  \coloneqq \hspace{2mm}   \begin{array}{c}\Qcircuit @R=2mm @C=2mm{
                         & \multigate{1}{V^{(M)}_{-}} & \qw &  \multigate{1}{V^{(M)}_{+}} & \gate{G} & \qw \\
                    & \ghost{V^{(M)}_{-}} & \gate{\tilde S^{(M)}} &  \ghost{V^{(M)}_{+}}  & \qw & \qw
        }\end{array}
            \end{align}
and hence $ \mathcal{O}_{\tilde{\mathcal{L}}}$ implements some $\tilde L_\nu$ in the following manner:
  \begin{align}
   \mathcal{O}_{\tilde{\mathcal{L}}}  (\ket{0} \otimes \ket{\psi}) =  \sum_{\nu} \ket{\nu} \otimes \tilde L_\nu \ket{\psi}.
    \end{align}
 
Let's find an expression for the $\tilde L_\nu$. We begin by observing that:
    \begin{align}
        V^{(M)}_{\pm} = \sum_\nu \sum_{x} \ketbra*{\nu \pm {m_j}^{(M)}}{\nu} \otimes G^\dagger_{\tilde P_x}G_{\tilde P_x} = \sum_\nu \sum_{x} \ketbra*{\nu \pm {m_j}^{(M)}}{\nu} \otimes \tilde P_x.
    \end{align}
    That way: 
    \begin{align}
        \mathcal{O}_{\tilde{\mathcal{L}}}(\ket{0}\otimes\ket{\psi})   &=\sum_{x,y}   \sqrt{G(m_x^{(M)} - m_y^{(M)})} \ket{m_x^{(M)} - m_y^{(M)}} \otimes \tilde{P}_x  \tilde S^{(M)} \tilde{P}_y \ket{\psi}\\
        &=\sum_\nu   \sqrt{G(\nu)} \ket{\nu} \otimes \sum_{\substack{x,y \\ m_x^{(M)} - m_y^{(M)} = \nu }} \tilde{P}_x \tilde S^{(M)} \tilde{P}_y \ket{\psi}.
    \end{align}

We see that:
   \begin{align} 
   \tilde L_\nu =  \sqrt{G(\nu)} \sum_{\substack{x,y \\ m_x^{(M)} - m_y^{(M)} = \nu }} \tilde{P}_x \tilde S^{(M)} \tilde{P}_y.
 \end{align}
 Observing $|G(\nu)| \leq 1$ and recalling the definition of $S(\omega)$, all that is left to do is to leverage \cref{prop:energyestimationgivenaroundingpromise} and \cref{lemma:distance-s-tildes} to bound the distance between $\tilde L_\nu$ and $L_\nu$.

 First observe that $\norm*{\tilde P_x} \leq 1$ which follows from the fact that $\tilde P_x$ is obtained from singular value transformation. By \cref{prop:lindbladsimulation}, we have

 \begin{align}
   &\quad \norm{\sum_{\substack{x,y \\ m_x^{(M)} - m_y^{(M)} = \nu }} P_x S^{(M)} P_y - \sum_{\substack{x,y \\ m_x^{(M)} - m_y^{(M)} = \nu }} \tilde P_x P^{(M)} S^{(M)} P^{(M)}\tilde P_y}\\ 
   &=\norm{\sum_{\substack{x,y \\ m_x^{(M)} - m_y^{(M)} = \nu }} P_x S^{(M)} P_y -  \tilde P_x P^{(M)} S^{(M)} P_y + \tilde P_x P^{(M)} S^{(M)} P_y - \tilde P_x P^{(M)} S^{(M)} P^{(M)} \tilde P_y} \\
   &\leq 2(s^{(M)})^2\delta_\mathrm{est}.
 \end{align}
Note that $P^{(M)}S^{(M)}P^{(M)} = S^{(M)}$. The above inequality implies that
 \begin{align}
   \label{eq:bound-1}
   \norm{\sum_{\substack{x,y \\ m_x^{(M)} - m_y^{(M)} = \nu }} P_x S^{(M)} P_y - \sum_{\substack{x,y \\ m_x^{(M)} - m_y^{(M)} = \nu }} \tilde P_x  S^{(M)} \tilde P_y} \leq 2(s^{(M)})^2\delta_\mathrm{est}.
 \end{align}
 Now, using \cref{lemma:distance-s-tildes}, we obtain
 \begin{align}
   &\quad \norm{\sum_{\substack{x,y \\ m_x^{(M)} - m_y^{(M)} = \nu }} \tilde P_x S^{(M)} \tilde P_y - \sum_{\substack{x,y \\ m_x^{(M)} - m_y^{(M)} = \nu }} \tilde P_x \tilde{S}^{(M)} \tilde P_y} \\
   &= \norm{\sum_{\substack{x,y \\ m_x^{(M)} - m_y^{(M)} = \nu }} \tilde P_x (S^{(M)}-\tilde S^{(M)}) \tilde P_y}  \\
   \label{eq:bound-2}
   &\leq (s^{(M)})^2\delta_\mathrm{leak}.
 \end{align}

Combining \cref{eq:bound-1,eq:bound-2}, we have
 \begin{align}
   \norm{L_\nu - \tilde L_\nu} \leq (s^{(M)})^2(2\delta_\mathrm{est} +  \delta_\mathrm{leak}).
 \end{align}

 So we have the desired property with $(s^{(M)})^2(\delta_{\mathrm{leak}} + 2\delta_{\mathrm{est}})$. Letting $\delta_L \coloneqq (s^{(M)})^2(\delta_\mathrm{leak} + 2\delta_\mathrm{est})$ and invoking \cref{lemma:jumpapprox} with $m = (s^{(M)})^2$, we have

\begin{align}
  \norm*{e^{\mathcal{L}t} - e^{\tilde{\mathcal{L}}t}}_{\diamond} &\leq t\norm*{\mathcal{L} - \mathcal{\tilde{L}}}_{\diamond} \\
  &\leq t \cdot\left[ 4(s^{(M)})^2(\delta_\text{leak}+2\delta_\text{est}) + 2(\delta_\text{leak} + 2\delta_\text{est})^2 \right]\\
  &= 16 (s^{(M)})^2t (\delta_\text{leak}+2\delta_\text{est}) ,
  %&= 4mt\norm{S}(\delta_\text{proj}+2\delta_\text{est}) + 2t(\delta_\text{proj} + 2\delta_\text{est})^2,
\end{align}
where the first inequality follows from the observation that for all integers $k \geq 0$,
\begin{align}
  \norm*{e^{\mathcal{L}t} - e^{\tilde{\mathcal{L}}t}}_{\diamond} = \norm*{(e^{\mathcal{L}t/k})^k - (e^{\tilde{\mathcal{L}}t/k})^k}_{\diamond} \leq k\norm*{e^{\mathcal{L}t/k} - e^{\tilde{\mathcal{L}}t/k}}_{\diamond} \leq \norm*{\mathcal{L}-\tilde{\mathcal{L}}}_{\diamond} + O(t/k).
\end{align}
Here the first inequality follows from the subadditivity of the diamond norm~\cite[Propositon 3.48]{Wat18}, and the last inequality is due to Taylor expansion.

We assume $S$ is given as a block encoding with normalizing constant 1 since $\norm{S}=1$. To use \cref{prop:lindbladsimulation}, it suffices to set
\begin{align}
  \delta_{\mathrm{leak}} = O\left(\frac{\delta_{\mathcal{L}}}{(s^{(M)})^4t}\right), \quad \text{ and } \delta_{\mathrm{est}} = O\left(\frac{\delta_{\mathcal{L}}}{(s^{(M)})^4t}\right)
\end{align}
to $\delta_{\mathcal{L}}$-approximately simulate $e^{\mathcal{L}^{(M)}t}$ by simulating $e^{\tilde{\mathcal{L}}^{(M)}t}$. Observe that $\norm{L_\nu}\leq 1$. Let $k$ be the number of system qubits. The simulation algorithm costs
\begin{align}
  \label{eq:gc-L}
  O\left((s^{(M)})^2t\,\frac{\log (s^{(M)}t/\delta_{\mathcal{L}})}{\log\log (s^{(M)}t/\delta_{\mathcal{L}})}\right)
\end{align}
queries to $\mathcal{O}_{\tilde{\mathcal{L}}}$ and
\begin{align}
  \label{eq:total-gc}
  O\left((s^{(M)})^4t\,\left(\frac{\log (s^{(M)}t/\delta_{\mathcal{L}})}{\log\log (s^{(M)}t/\delta_{\mathcal{L}})}\right)^2\right)
\end{align}
additional 1- and 2-qubit gates. Note that the number of queries to $O_{\tilde{\mathcal{L}}}$ is also the number of innovations to $S$.

To implement $\mathcal{O}_{\tilde{\mathcal{L}}}$, we invoke \cref{prop:energyestimationgivenaroundingpromise} with precision parameter $\delta_{\mathrm{est}}$ and \cref{lemma:attenuationoperators} with precision parameter $\delta_{\mathrm{leak}}$. Each application of $\mathcal{O}_{\tilde{\mathcal{L}}}$ has the following number of invocations of $U_H$:
\begin{align}
  & O\left(   \kappa^{-1}  \log(s^{(M)})^2 \log(\delta_\mathrm{est}^{-1})  + \kappa^{-1} \gamma^{-1} \log(s^{(M)} \delta_\mathrm{leak}^{-1})   \right) \\
  =& O\left(   \kappa^{-1}  \log(s^{(M)})^2 \log( s^{(M)} t \delta_\mathcal{L}^{-1})  + \kappa^{-1} \gamma^{-1} \log(s^{(M)}) \log( s^{(M)} t \delta_\mathcal{L}^{-1})    \right) \\
  =&O\left(    \kappa^{-1} \gamma^{-1}  \log(s^{(M)})^2 \cdot \log( s^{(M)} t \delta_\mathcal{L}^{-1})   \right).
\end{align}
\end{proof}

To achieve the final number of queries to the block encoding of $H$, we observe that the rounding promises $M_j$ have $\kappa^{-1} \in O( 2^{n+r} )$ and $s^{(M)} \in O(2^n)$. Plugging these in, and accounting for the cost of the left-right POVM from \cref{lemma:leftrightprojectionoperator} we obtain:
\begin{align}
    O\left( \gamma^{-1} \cdot n^2 2^{3n+r}  t \cdot \mathrm{polylog}(  t/\delta_\mathcal{L} ) \right).
\end{align}
This establishes \cref{claim:implementingdynamics}.
Now we recall \cref{thm:accuracyofthefinalensemble} which states that the final output state $\rho^{*M}$ satisfies:
\begin{align}
     \norm{ \rho^{*M} - \rho_\beta  }_1 \leq  \sqrt{\beta 2^{-n}} + 2\cdot 2^{-r}.
\end{align}
To achieve $ \norm*{ \rho^{*M} - \rho_\beta  }_1 \leq  \varepsilon$, we select $n = \log_2(\beta  (\varepsilon/2)^{-2})$ and $r = \log_2(4/\varepsilon)$. Plugging these into the query complexity, we get:
\begin{align}
     O\left( \gamma^{-1} \cdot t \cdot \beta^3 \varepsilon^{-7}  \cdot \mathrm{polylog}(  t/\delta_\mathcal{L}) \cdot \log^2(\beta/\varepsilon) \right).
\end{align}
This establishes \cref{thm:mainresult}.

%%%

\section{Some open questions}\label{sec:open_questions}

Our result achieved a time complexity that scales linearly in the mixing time $t_\text{mix}$. However, the performance with respect to the inverse temperature $\beta$ and accuracy $\varepsilon$ is $\tilde O(\beta^3 \varepsilon^{-7})$ which has plenty of room for improvement. One potential path that may yield an accuracy of $\tilde O(\beta \varepsilon^{-2})$ is to remove the linear dependence on the number of jump operators in the Lindblad simulation algorithm from \cref{prop:lindbladsimulation}. \cref{prop:lindbladsimulation} demands block encodings of the individual jump operators $L_j$, but the oracle $\mathcal{O}_\mathcal{L}$ we prepare may actually be much more powerful than this. To see this, first note that if we are given access to the isometry $\sum_j\ket{j}\otimes K_j$, then implementing the channel is trivial as we already have the Stinespring dilation. In our algorithm, we implemented an oracle in the form of $\sum_j \ket{j}\otimes L_j$, which is close to the Stinespring dilation we want because for the infinitesimal approximation channel in~\cite{1612.09512}, all but one Kraus operators are proportional to $L_j$. That one special Kraus operator involves all the $L_j$'s. Does there exist any special treatment of this special Kraus operator so that we can leverage the special structure of the oracle $\sum_j \ket{j}\otimes L_j$ to get rid of the $O(m)$ dependence?\footnote{As remarked earlier, after the first version of this manuscript was made public, recent work \cite{2303.18224} resolved this question in the affirmative.}

A central goal in our work is to attain a rigorous bound on the accuracy of the final output state. To this end, we assume that we are given a lower bound on the mixing time $t_\text{mix}$, so that we know for how long the dynamics of the Davies generator must be simulated to achieve a high-accuracy output state. But this assumption is rarely the case in practice. Furthermore, the attenuation discussed in \cref{sec:lindbladdynamicsonthepromisedsubspace} may slow the mixing time and exasperate this problem. Is there a technique that can detect if the Davies generator has been run for long enough?  One approach might be to purify the dynamics of the Lindblad simulation, and then use amplitude estimation to compare the resulting pure output state to an ideal thermal state purification. A reflection operator around a purification of the thermal state could be obtained via the techniques from \cite{2107.07365}. 

The only piece of our method that eludes rigorous mathematical treatment is the impact of attenuation on mixing time $\tilde t_\text{mix}$. Certainly the dependence of $\gamma^{-1}$ in the circuit complexity the synthesis attenuation operators $A^{(M)}$ is optimal, due to lower bounds on approximation of threshold functions with polynomials. But perhaps given additional knowledge about the Hamiltonian, there may be other approaches for selecting coupling operators $S_\alpha$ that do not leak out of a given promised subspace.

\section{Acknowledgements}

The authors thank Kristan Temme and Chi-Fang Chen for helpful discussions, as well as anonymous reviewers for their helpful comments. CW was supported by a seed grant from the Institute of Computational and Data Science (ICDS) and a National Science Foundation grant CCF-2238766 (CAREER).

\appendix

\section{Glossary\label{app:glossary}}

This manuscript features many quantities and mathematical symbols, so give a brief description of some of these with references to the relevant part of the text.
\begin{description}
\item[Hamiltonian.] $H$ is decomposed into eigenvalues $\lambda_i$ and eigenspace projectors $\Pi_i$ via $\sum_i \lambda_i \Pi_i.$ Our goal is the prepare the  thermal state $\rho_\beta$ at inverse temperature $\beta$. See \cref{def:hamiltonian}.
\item[Rounding promise.] $M \subset [0,1]$ defines a promised subspace $\mathcal{P}^{(M)}$ and a projector onto that subspace $P^{(M)}$. $M$ consists of intervals $[a_x,b_x]$ with corresponding promised eigenspace projectors $P^{(M)}_x$. See \cref{def:roundingpromise}. We can perform energy estimation with respect to approximate promised eigenspace projectors $\tilde P^{(M)}_x$ using \cref{prop:energyestimationgivenaroundingpromise}. We suppress the superscript $(M)$ when a promise is clear from context.  
\item[Lindbladian.] $\mathcal{L}$ is a superoperator defining a continuous time quantum Markov process describing open system dynamics. It is defined by the jump operators $L_{\omega,\alpha}$. See  \cref{eq:lindblad}.
\item[Davies Generator.] A particular choice of jump operators $L_{\omega,\alpha} = \sqrt{G(\omega)} S_\alpha(\omega)$ that yields thermalizing dynamics.  $\omega$ is among the Bohr frequencies of $H$: the set of pairwise energy differences. The filter function $G(\omega)$ biases the dynamics towards certain energy differences, and the coupling operators $S_\alpha$ `scramble' the Hilbert space. See \cref{def:daviesgenerator}.
\item[Initial state.] Our simulation of thermalizing dynamics starts in an arbitrary initial state $\sigma$. After measuring the left-right POVM defined by $P_\mathrm{LR}$, we obtain an initial state $\tilde \sigma^{(M)}$ satisfying a rounding promise $M \in \{\bar L,\bar R\}$. See \cref{sec:leftrightpovm}.
\item[Attenuation.] We require coupling operators $S^{(M)}$ that do not `leak' out of the promised subspace $\mathcal{P}^{(M)}$. We achieve this by sandwiching them between attenuation operators $A^{(M)}$ with the same support as $P^{(M)}$, but also some eigenvalues in $\mathcal{P}^{(M)}$ are attenuated to be less than 1 which slows the mixing time. The attenuation coefficient $\gamma$ controls the number of these eigenvalues, but an approximate implementation $\tilde A^{(M)}$ of $A^{(M)}$ requires circuit complexity $O(\gamma^{-1})$. See \cref{sec:lindbladdynamicsonthepromisedsubspace}.
\item[Promised thermal states.] When the Hamiltonian $H$ is truncated onto $\mathcal{P}^{(M)}$ we obtain the well-rounded Hamiltonian $H^{(M)}$ with eigenvalues $m_x^{(M)}$ and eigenspace projectors $P_x^{(M)}$, see \cref{remark:wellroundedhamiltonian}. On $\mathcal{P}^{(M)}$, it has the promised thermal state $\rho^{(M)}_\beta$ and promised partition function $\mathcal{Z}_\beta^{(M)}$, see \cref{def:promisedgibbs}.  In section \cref{sec:ensembleanalysis}, it is convenient to define `exact' versions of these  $\hat{\rho}^{(M)}_\beta,\hat{\mathcal{Z}}_\beta^{(M)}$ that are still truncated to $\mathcal{P}^{(M)}$ but are based on the eigenvalues $\lambda_i$ of $H$.
\item[Mixing time.] We assume that the ideal Davies generator $\mathcal{L}$ with jump operators $L_{\omega,\alpha} = \sqrt{G(\omega)} S_\alpha(\omega)$ requires the mixing time $t_\mathrm{mix}$ in order to transform an arbitrary input state into something close to $\rho_\beta$. Due to attenuation, the promised Davies generator $\mathcal{L}^{(M)}$ with jump operators $L_{\omega,\alpha} = \sqrt{G(\omega)} S^{(M)}_\alpha(\omega)$ will have a slower mixing time $\tilde t_\mathrm{mix}$. See \cref{def:promiseddaviesgenerator}.
\end{description}

\section{Polynomial construction\label{app:polynomialconstruction}}

Throughout the paper, we required the construction of block encodings.  \cref{prop:energyestimationgivenaroundingpromise} gave operators $P_x$ that indicate the eigenspaces of a particular interval $[a_x,b_x]$ of a rounding promise $M$, and a unitary $U^{(M)}$ that computes the binary expansion $\ket{x}$ accordingly. \cref{lemma:leftrightprojectionoperator} gave an operator $P_\mathrm{LR}$ can be used to force either the $\bar L$ or $\bar R$ rounding promise via POVM, by making $P_\mathrm{LR}$ small in $\mathcal{P}^{(\bar L)}$ and large in $\mathcal{P}^{(\bar R)}$. Finally, \cref{lemma:attenuationoperators} gave an attenuation operator $A^{(M)}$ that vanishes outside of $\mathcal{P}^{(M)}$, but is simultaneously as large as possible within $\mathcal{P}^{(M)}$.

These operators $P_x, P_\mathrm{LR},$ and $A^{(M)}$ have a lot in common: they all commute with the Hamiltonian, and can hence be seen as $\sum_i f(\lambda_i) \Pi_i$ for some function $f$. Furthermore, the requirements on $f$ are always that $f(\lambda) = 0$ or $f(\lambda) = 1$ for $\lambda$ in certain regions. We construct all of these operators through singular value transformation, which lets us construct such operators via polynomial approximations of $f$. 

The requirements on the operators ensure that there are always gaps between the regions of $\lambda$ where $f(\lambda) = 0$ or $f(\lambda) = 1$. This is essential for polynomial approximation of $f$, since polynomials are always continuous. The degree of the polynomial scales with the reciprocal of the width of the gap. Our starting point for the polynomial construction is a highly accurate polynomial approximation of a step function.

\begin{lemma}[\textbf{Polynomial approximation of a step function} ({\cite[Appendix A]{1707.05391}})]
  \label{lemma:polynomialapproximationofastepfunction} 
  For any $\kappa, \delta >0$, there exists an odd polynomial $\Theta_{\kappa,\delta}$ of degree $O(\kappa^{-1} \log(\delta^{-1}))$ such that $\forall \lambda \in [-2,2]$ we have  $0 \leq \Theta_{\kappa,\delta}(\lambda) \leq 1$ and:
\begin{align}
  \text{if } \lambda \leq -\kappa/2,  &\text{ then } \Theta_{\kappa,\delta}(x) < \delta; \\
  \text{if } \kappa/2 \geq \lambda,  &\text{ then }  1-\delta \leq \Theta_{\kappa,\delta} (\lambda).
\end{align}
\end{lemma}

Given a polynomial approximation of a single transition from $f(\lambda) = 0$ to $f(\lambda) = 1$, we can approximate an arbitrary sequence of transitions by shifting and adding multiple such polynomials together. Above, we guaranteed that the approximation of the step function holds for $\lambda \in [-2,2]$, so that we can shift the polynomial by up to $1$ and still have good approximation on the interval $[-1,1]$.

\begin{lemma}[\textbf{Construction of approximate projection polynomials}]
  \label{lemma:constructionofapproximateprojectionpolynomials} 
  Say we have a collection of intervals $\{[a_x,b_x]\}$ in $[0,1]$ with $a_x < b_x$ and $b_x < a_{x+1}$, and each interval is labeled with a bit $c_x \in\{0,1\}$. Say furthermore that each of the intervals is at least $\kappa$ far apart, that is, $b_x + \kappa < a_{x+1}$. 

Then, for any $\delta > 0$ there exists a polynomial $P(\lambda)$ such that:
\begin{align}
    \forall x, \text{ if }\,\forall \lambda\in [a_x,b_x], \text{ then } |P(\lambda) - c_x| \leq \delta. 
\end{align}
Furthermore, the polynomial satisfies $0 \leq P(\lambda)\leq 1$. Say there are $K$ many `flips', that is, indexes $x$ where $c_x \neq c_{x+1}$. Then the polynomial is of degree $O(\kappa^{-1} \log(K\delta^{-1}))$.
\end{lemma}
\begin{proof} We will construct $P(\lambda)$ by adding together several $\Theta_{\kappa_x,\delta/K}$ for various values of $\kappa_x$, all of which satisfy $\kappa_x\geq\kappa$. Since adding several polynomials doesn't change the degree, the resulting polynomial $P(\lambda)$ has degree $O(\kappa^{-1}\log(K\delta^{-1}))$.

There will be one $\Theta_{\kappa_x,\delta/K}$ for each flip, that is, each pair of intervals with $c_x \neq c_{x+1}$. Let $t_x \coloneqq (b_x+a_{x+1})/2$ be the midpoint between two intervals, and let $\kappa_x = a_{x+1}-b_{x} \geq \kappa $ be the distance between them. We construct:
\begin{align}
    P(\lambda) \coloneqq  c_1 + \sum_{\substack{x\\ c_x \neq c_{x+1}}} \Theta_{\kappa_x, \delta/K}(\lambda-t_x) \cdot \Bigg\{ \begin{array}{cc}
       +1  & \text{if } c_x < c_{x+1} \\
        -1  & \text{if } c_x > c_{x+1}
 \end{array}.
\end{align}

It remains to prove that $P(\lambda)$ satisfies the desired property. By \cref{lemma:polynomialapproximationofastepfunction} and the choice of $t_x,\kappa_x$ it is guaranteed that $\Theta_{\kappa_x,\delta/K}(\lambda-t_x)$, is always either $\leq \delta/K$ or $\geq 1-\delta/K$ outside of the region $[b_x, a_{i+x}]$. So, for any interval $[a_x,b_x]$ we have that \text{any} of the $\Theta_{\kappa_x,\delta/K}(\lambda-t_x)$ is within $\delta/K$ of $0$ or $1$. Thus, for the purposes of analyzing $P(\lambda)$, let us pretend for the rest of the proof that they are \emph{exactly} 0 or 1. In doing so, we will be wrong by at most $K \cdot \delta/K = \delta$.

Let us consider any particular interval $[a_x,b_x]$. We want to argue that for any $\lambda$ in this interval, $P(\lambda) = c_i$. We achieve this through induction in $K$. The base case is very easy: if $K = 0$, then all the $c_x$ are equal, so we can just set $P(\lambda) = c_1$.

Now, consider only the first $K-1$ switches, and let $P'(\lambda)$ be the polynomial only from these. If $c_x$ is the last $x$ such that $c_{x} \neq c_{x+1}$, then $P'(\lambda)$ satisfies $\lambda \in [a_y, b_y] \implies P'(\lambda) = c_y$ for $y \leq x$, and $\lambda \in [a_y, b_y] \implies P'(\lambda) = 1-c_y$ for $y > x$. The construction of $P(\lambda)$ depends on if $c_{x+1} = 0$ or $c_{x+1} = 1$. If $c_{x+1} = 0$, then $\lambda \in [a_{x+1}, b_{x+1}] \implies P'(\lambda) = 1$, so we need to subtract $1$ for $\lambda > b_x$, which is achieved by subtracting $\Theta_{\kappa_x,\delta/K}(\lambda-t_x)$ from $P'(\lambda)$. Otherwise, if $c_{x+1} = 1$, then $\lambda \in [a_{x+1}, b_{x+1}] \implies P'(\lambda) = 0$, so we need to add $1$ for $\lambda > b_x$, which is achieved by adding $\Theta_{\kappa_x,\delta/K}(\lambda-t_x)$. The conditions $c_{x+1} = 0$ and $c_{x+1} = 1$ are equivalent to $c_x > c_{x+1}$ and $c_x < c_{x+1}$ respectively. Consequently:
\begin{align}
    P(x) = P'(x) +  \Theta_{\kappa_x, \delta/K}(x-t_x) \cdot \Bigg\{ \begin{array}{cc}
       +1  & \text{if } c_x < c_{x+1} \\
       -1  & \text{if } c_x > c_{x+1}\end{array}.
\end{align}

\end{proof}

It remains to invoke singular value transformation in order to construct operators with this polynomial as their spectrum. The polynomial we have constructed has mixed parity, which is acceptable because we perform singular value transformation on a Hermitian operator. The method from \cite{1806.01838} splits the polynomial into its even and odd parts and combines them together via a linear combination of block encodings. This introduces an extra factor of $\frac{1}{2}$. Since we only care about making eigenvalues either close to 0 or 1, we can use a simple version of oblivious amplitude amplification via the Chebyshev polynomial $T_3$ to remove this extra factor. 

\begin{lemma}[\textbf{Construction of projectors via singular value transformation}]
  \label{lemma:constructionofprojectorsviasingularvaluetransformation} 
  Say $H$ is a Hermitian matrix with $-1 \leq H \leq 1$, and we are given a block encoding $U_H$ of $H$. Say $H$ has the eigendecomposition $H = \sum_i \lambda_i \Pi_i$.
Say $P(\lambda)$ is a degree-$d$ polynomial that, for some $\delta$ and for certain regions $\{[a_x,b_x]\}$ satisfies $\lambda\in[a_x,b_x] \implies |P(\lambda) - c_x| \leq \delta$ for $c_x \in \{0,1\}$, and furthermore satisfies $0 \leq P(\lambda) \leq 1$ for all $\lambda\in [-1,1]$.

Then, there exists a quantum circuit $U_{\bar P(H)}$ which is a block encoding of a matrix $\bar P(H)$ defined by:
\begin{align}
    \bar P(H) \coloneqq \sum_i \bar P(\lambda_i) \Pi_i.
\end{align}
where $\bar P$ is a function that satisfies $\lambda\in[a_x,b_x] \implies |\bar P(\lambda) - c_x| \leq 2\delta$. This circuit makes $d$ many uses of controlled-$U_H$, and has circuit complexity $O(d)$ overall.
\end{lemma}
\begin{proof} Our starting point is Theorem~56 of \cite{1806.01838}, which allows us to construct a block encoding of
\begin{align}
    \frac{1}{2}P(H) \coloneqq \sum_i \frac{1}{2}P(\lambda_i) \Pi_i.
\end{align}
To get rid of the factor of $\frac{1}{2}$, let $T_{3}(\lambda)$ be the third Chebyshev polynomial of the first kind, and use Corollary~18 of \cite{1806.01838} to construct a block encoding of:
\begin{align}
    -T_{3}\left(\frac{1}{2}P(H)\right) \coloneqq \sum_i -T_3\left(\frac{1}{2}P(\lambda_i)\right) \Pi_i.
\end{align}
Observe that $-T_3(0) = 0$ and $-T_3(\frac{1}{2}) = 1$, and also that $-T_3(\delta/2) \leq 2\delta$ and $-T_3((1-\delta)/2) \geq 1 - 2\delta$. So if we let $\bar P(\lambda) \coloneqq -T_3\left( \frac{1}{2}P(\lambda)\right)$ then $ -T_{3}\left(\frac{1}{2}P(H)\right)$ is the $\bar P(H)$ we wanted to construct. Since $P(\lambda)$ has degree $d$, and $-T_3(\lambda)$ has degree $O(1)$, the total circuit complexity and number of invocations to controlled-$U_H$ is $O(d)$ overall.
\end{proof}

Now we are ready to synthesize the operators we need in order to prove \cref{prop:energyestimationgivenaroundingpromise}, \cref{lemma:leftrightprojectionoperator}, and \cref{lemma:attenuationoperators}.

We start with the energy estimation result from \cref{prop:energyestimationgivenaroundingpromise}, which is adapted from \cite{2103.09717}. This work gives a construction for energy estimation that attempts to minimize the constant factors in polynomial degree as well as ancilla count. We do not concern ourselves with these, and just give an asymptotic bound that also has an extra factor of $n$ in the complexity. A similar method also appears in \cite{2105.02859}.

\begin{proof}[Proof of \cref{prop:energyestimationgivenaroundingpromise}.] Say $U_A$ is a block encoding with $k$ ancilla qubits of a Hermitian operator $A$. For $y \in \{0,1\}^k$ we define $A_y \coloneqq (\bra{y} \otimes I)U_A (\ket*{0^k} \otimes I)$, so that:
\begin{align}
U_A \left(\ket*{0^k} \otimes I\right) = \sum_y \ket{y} \otimes A_y.
\end{align}

Then, let $G_A \coloneqq \ket*{0^k} \otimes A$ and $\bar G_A := \sum_{y \in \{0,1\}^k\setminus\{0^k\}} \ket{y} \otimes A_y$ and apply a generalized Toffoli gate to synthesize the isometry:
\begin{align}
 & \left(\ket{1}\otimes \ketbra*{0^k}{0^k} + \ket{0}\otimes \sum_{y \in \{0,1\}^k\setminus\{0^k\}} \ketbra{y}{y}\right) U_A \left(\ket*{0^k} \otimes I\right) \\
 = & \ket{1} \otimes G_A + \ket{0} \otimes \bar G_A.
\end{align}

Clearly $G_A^\dagger G_A = A^2$. We also have $\bar G_A^\dagger \bar G_A = \sum_{y \in \{0,1\}^k\setminus\{0^k\}} A_y^\dagger A_y = I - A^2$ following from the unitarity of $U_A$.

Our goal is to compute $\ket{x}$, where $x$ is the index of the interval $[a_x,b_x]$ containing $\lambda$. We will repeatedly use the above isometry to compute $\ket{x}$ one bit at a time. To do so, we make use of \cref{lemma:constructionofapproximateprojectionpolynomials,lemma:constructionofprojectorsviasingularvaluetransformation} to synthesize $\delta_j$-accurate projectors for each bit of $x$ --- call them $\tilde P^{(j)}$ for the approximate projector for the $j$'th bit. 

If we apply the above construction involving $U_A$ above for each of $\tilde P^{(j)}$, we construct an isometry $\tilde E^{(M)}$ that achieves $\tilde E^{(M)} = \sum_x \ket{x} \otimes G_{\tilde P_x}$ with the desired properties. We observe that:
\begin{align}
    G_{\tilde P_x} \coloneqq \prod_{j = 1}^{n+1} \Bigg\{ \begin{array}{c} G_{\tilde P^{(j)}}  \text{ if } x_j = 1 \\ \bar G_{\tilde P^{(j)}}  \text{ if } x_j = 0 \end{array}.
\end{align}
Since the $\tilde P^{(j)}$ commute, we can evaluate:
\begin{align}
\tilde P_x := G_{\tilde P_x}^\dagger G_{\tilde P_x} = \prod_{j = 1}^{n+1} \Bigg\{ \begin{array}{c} G_{\tilde P^{(j)}}^\dagger G_{\tilde P^{(j)}}  \text{ if } x_j = 1 \\ \bar G_{\tilde P^{(j)}}^\dagger \bar G_{\tilde  P^{(j)}}  \text{ if } x_j = 0 \end{array} =  \prod_{j = 1}^{n+1} \Bigg\{ \begin{array}{c} (\tilde P^{(j)})^2 \text{ if } x_j = 1 \\ I -  (\tilde P^{(j)})^2  \text{ if } x_j = 0 \end{array}.  
\end{align}

It remains to show that $\left\|\tilde P_x P^{(M)}  - P_x \right\| \leq \delta_\text{est}$.  The $j$-th bit has $2^j$ many flips. We take $\delta_j \coloneqq  \delta_\text{est}/2 \cdot 2^{-j} $ such that $\left\| \tilde P^{(j)} P^{(M)}  - P^{(j)} \right\| \leq \delta_j$, with $P^{(j)}$ being the ideal projector onto the $j$'th bit. That way the total error in spectral norm is at most $\delta_\text{est} \cdot \sum_{j=1}^{n+1} 2^{-j} \leq \delta_\text{est}$. Each projector has complexity $O\left(\kappa^{-1} \log(2^{j} \delta_j^{-1})\right)$. The total complexity is:
\begin{align}
 \kappa^{-1}  \sum_{j=1}^{n+1} \log(2^{j} \delta_j^{-1})   = \kappa^{-1}  \sum_{j=1}^{n+1} \log(4^{j} \cdot 2\delta_\text{est}^{-1} ) = O\left( \kappa^{-1}  n^2 \log(\delta_\text{est}^{-1}) \right).
\end{align}
\end{proof}

Next, we construct the operator that underlies the left-right POVM. To follow this argument, we refer again to \cref{fig:roundingpromises} which gives a sketch of $\bar L, \bar R$ as well as the function approximated by \cref{lemma:constructionofapproximateprojectionpolynomials}.

\begin{proof}[Proof of \cref{lemma:leftrightprojectionoperator}] The operator $P_\mathrm{LR}$ is constructed using \cref{lemma:constructionofapproximateprojectionpolynomials,lemma:constructionofprojectorsviasingularvaluetransformation}: all we need to do is select the intervals $[a_x,b_x]$ and the labels $c_x$. Recall that the goal was to ensure that $\|P_\mathrm{LR}^2\Pi_i\| \leq \delta_\text{sup}$ for $\lambda_i \not\in \bar L$ and $\|(I-P_\mathrm{LR}^2)\Pi_i\| \leq \delta_\text{sup}$ for $\lambda_i \not\in \bar R$. Recall \cref{fig:roundingpromises}. The intervals of interest are the \emph{gaps} of $\bar L$ and $\bar R$, for which we set $c_x = 0$ and $c_x = 1$ respectively. If we invoke \cref{lemma:constructionofapproximateprojectionpolynomials} with precision $\delta = \delta_\text{sup}/6$, then  for $\lambda_i \not\in \bar L$ we have  $\|P_\mathrm{LR}^2\Pi_i\| \leq \delta^2 \leq \delta_\text{sup}/3 $ and for $\lambda_i \not\in \bar R$ we have  $\|(I-P_\mathrm{LR}^2)\Pi_i\| \leq 1- (1-\delta)^2 \leq \delta_\text{sup}/3 $ as desired.

From the construction of $\bar L$ and $\bar R$ in \cref{def:thefinegrainedroundingpromises}, we see that each of the intervals are exactly $2^{-n-r-2}$ apart, and there are $2^{n+r+1}$ of them. Since the intervals alternate with $c_x = 0$ and $c_x = 1$, there are $O(2^{n+r})$ many switches. So, the implementation requires $(n+r)2^{n+r}\log(\delta^{-1}_\text{sup})$ many invocations of the block encoding of $H$.
\end{proof}

Finally, we construct the attenuation operator $A^{(M)}$ and its approximation $\tilde A^{(M)}$. We demand that $A^{(M)}$ vanishes outside of the rounding promise $M$, and is close to 1 in inside of a truncated rounding promise $M^T$. For all the other eigenspaces, the approximation $\tilde A^{(M)}$ and $A^{(M)}$ agree. This means that we can use the construction of $\tilde A^{(M)}$ to define what $A^{(M)}$ should be outside these regions.  

\begin{proof}[Proof of \cref{lemma:attenuationoperators}] As usual, we use \cref{lemma:constructionofapproximateprojectionpolynomials,lemma:constructionofprojectorsviasingularvaluetransformation} to construct the block encoding of $\tilde A^{(M)}$. This construction results in a function $\tilde a^{(M)}$ which lets us define $a^{(M)}$ and hence $A^{(M)}$ via the requirements in \cref{lemma:attenuationoperators}.

We need to ensure that $\tilde a^{(M)}(\lambda) \leq \delta_\text{leak}$ for $\lambda \not\in M$, and that $\tilde a^{(M)}(\lambda) \geq 1-\delta_\text{leak}$ for $\lambda \in M^T$. This immediately shows what intervals to use for \cref{lemma:attenuationoperators}: the intervals with $c_x = 0$ are the gaps of $M$, and the intervals with $c_x = 1$ are the intervals of $M^T$.

There are $2$ `flips' per interval of $M$, so there are $O(s^{(M)})$ flips total. Since $M^T$ was obtained by taking $M$ and shrinking the intervals by $\kappa\gamma$ on each side, the gap width for the purposes of \cref{lemma:attenuationoperators} is $\kappa\gamma$. Accordingly, the polynomial degree and hence the query complexity are $O\left(\kappa^{-1} \gamma^{-1} \log( s^{(M)} \delta^{-1} ) \right)$.

\end{proof}

\section{The Approximate Lindbladian\label{app:approximateDaviesgenerator}}

We describe here a method that does not try to force any rounding promises to make energy estimation unambiguous.  Rather this method works directly with approximate jump operators that arise from imperfect energy estimation of the Hamiltonian $H$ on the entire Hilbert space $\mathcal{H}$.

We consider a general energy estimation unitary $V^{(\pm)}$ that acts as
\begin{align}\label{eq:gen_ee}
  V^{(\pm)} &= \sum_i \Pi_i \otimes \sum_x f(\lambda_i,\tilde{\lambda}_x) \ketbra*{z}{z\pm \tilde{\lambda}_x},
\end{align}
where $\lambda_i$'s are the energies of $H$ and $\Pi_i$'s the projectors onto the corresponding eigensubspaces, and $\tilde{\lambda}_x$'s denote the possible outcomes produced by the energy estimation method. The summation over $x$ means that any energy estimation method necessarily produces superpositions of energy estimates. Typically, the magnitude of $f(\lambda_i,\tilde{\lambda}_x)$ decreases with increasing distance between the true energy $\lambda_i$ and the approximate energies $\tilde{\lambda}_x$.

For instance, when energy estimation is based on phase estimation the unitary $U=\exp(i H)$, the estimates are simply $n$-bit binary fractions $\tilde{\lambda}_x=x/2^n$ and the amplitudes $f(\lambda_i,\tilde{\lambda}_x)$ are given by
\begin{align}
    f(\lambda_i,\tilde{\lambda}_x)
    &=
    \frac{1}{2^n} \; 
    \frac{\exp\big(\pi i \; 2^n(\lambda_i - \tilde{\lambda}_x)\big)}{\exp\big(\pi i \; (\lambda_i - \tilde{\lambda}_x)\big)}.
\end{align}
The spread of $f(\lambda_i,\tilde{\lambda}_x)$ can be made narrower with the help of median amplification.

We now construct an oracle $\mathcal{O}_{\tilde{\mathcal{L}}}$ encoding jump operators of an approximate Lindbladian $\tilde{\mathcal{L}}$ using a general estimation unitary $V^{(\pm)}$ as in \cref{eq:gen_ee}.  We consider the case that there is only one coupling operator $S$. $\mathcal{O}_{\tilde{\mathcal{L}}}$ is given by the circuit:
\begin{align}
\mathcal{O}_{\tilde{\mathcal{L}}}  
\coloneqq \hspace{2mm}   
\begin{array}{c}\Qcircuit @R=2mm @C=2mm{
    & \multigate{1}{V^{(-)}} & \qw      &  \multigate{1}{V^{(+)}} & \gate{G} & \qw \\
    & \ghost{V^{(-)}}        & \gate{S} &  \ghost{V^{(+)}}  & \qw & \qw
}
\end{array}
\end{align}
where $G$ denotes a block encoding of:
\begin{align}
    G \coloneqq \sum_\nu \sqrt{G(\nu)} \otimes \ketbra{\nu}{\nu}.
\end{align}

%%%

To describe the jump operators that are encoded by the above circuit, we define the approximate Bohr frequencies $\nu$ to be differences of the form $\tilde{\lambda}_x-\tilde{\lambda}_y$.  The approximate jump operators $\tilde{L}(\nu)$ are given by 
\begin{align}
    \tilde{L}(\nu) &= \sqrt{G(\nu)} \, \tilde{S}(\nu),
\end{align}
where 
\begin{align}
    \tilde{S}(\nu) 
    &=
    \sum_{\substack{x,y\\ \tilde{\lambda}_x-\tilde{\lambda}_y = \nu}} 
    \sum_{i,j} 
    f(\tilde{\lambda}_x,\lambda_i)
    \overline{f(\tilde{\lambda}_x,\lambda_j}) 
    \Pi_i S \Pi_j \\
    &=
    \sum_{\substack{x,y\\ \tilde{\lambda}_x-\tilde{\lambda}_y = \nu}} A(x) S A(y)^\dagger.
\end{align}
Unfortunately, the operators 
\begin{align}
    A(x) &= \sum_{i} f(\tilde{\lambda}_x,\lambda_i) \Pi_i
\end{align}
that ``sandwich'' the coupling operator $S$ are not projectors as in the case for Davies generators.  Thus, it is not possible to interpret the operators $\tilde{L}(\nu)$ as jump operators of a Davies generator with respect to some Hamiltonian that is close to the original Hamiltonian.  Therefore, it is much more difficult to determine the fixed point of the corresponding Lindbladian $\tilde{\mathcal{L}}$, which is given by
\begin{align}
    \tilde{\mathcal{L}}(\rho)
    &=
    \sum_{\nu}
    G(\nu) \left[
        \tilde{S}(\nu) \rho \tilde{S}(\nu)^\dagger +
        \frac{1}{2} \left(
            \tilde{S}(\nu)^\dagger \tilde{S}(\nu) \rho +
            \rho \tilde{S}(\nu)^\dagger \tilde{S}(\nu)
        \right)
    \right].\label{eqn:apxDavies}
\end{align}

This `approximate Davies generator' $\tilde{\mathcal{L}}$ is challenging to analyze. It is intuitive that this Davies generator's steady state must be somewhat close to the ideal thermal state: the error due to the finite precision of energy estimation can be dealt with in the same way as for the promised Davies generators, and the `rounding errors' stemming from energies located at $x/2^n + 1/2^{n+1}$ also only shift the accuracy of the estimate slightly. 

To our knowledge, no method exists for rigorously proving a bound on the distance between this Lindblad operator's steady state and the true thermal state. A potential candidate for a proof technique via `approximate detailed balance' appeared in \cite{TVOVP11,2112.07646}, where it was leveraged to prove the accuracy of the quantum Metropolis algorithm. But it is not clear how to translate this technique to Davies generators.

How severe are these rounding errors? In practice, it may be the case that the steady state of approximate Davies generator is close to the ideal thermal state. Here, we present some evidence that it may be difficult to prove such a claim without an assumption on the Hamiltonian. We construct an `adversarial' Hamiltonian that places its eigenvalues exactly at the locations where the rounding errors are maximized. Specifically, its energies are located at:

\begin{align}
    \lambda_i = \frac{ i + \alpha/2  }{2^n} \label{eqn:adversarial}
\end{align}
where $0 \leq i < 2^n$ is an integer and $\alpha \in [0,1]$ is an `adversariality parameter'. When $\alpha = 0$, then energy estimation is perfect and the operators $A(x)$ are projectors exactly. But as $\alpha$ increases, the eigenvalues are shifted so that they are rounded up or down with increasing entropy. This Hamiltonian is specifically designed to capture these rounding errors alone: the energies are spaced out evenly with exactly the precision of the energy estimation protocol.

\cref{fig:adversarial} shows the error in preparing the ideal thermal state when using the approximate Davies generator with this adversarial Hamiltonian. We see that at $\alpha = 0$ the method computes the thermal state exactly. But as $\alpha$ increases we see significant errors. If we amplify the accuracy of the energy measurement using median amplification, then the errors appear only for larger $\alpha$. But no amount of amplification can remove the error for $\alpha = 1$.

Since the approximate Davies generator is significantly simpler and may be less costly to implement than our scheme of random promised Davies generators, a proof technique establishing rigorous accuracy bounds on this Davies generator may result in a substantially improved algorithm for thermal state preparation.

\begin{figure}
    \centering
    \includegraphics[width=\textwidth]{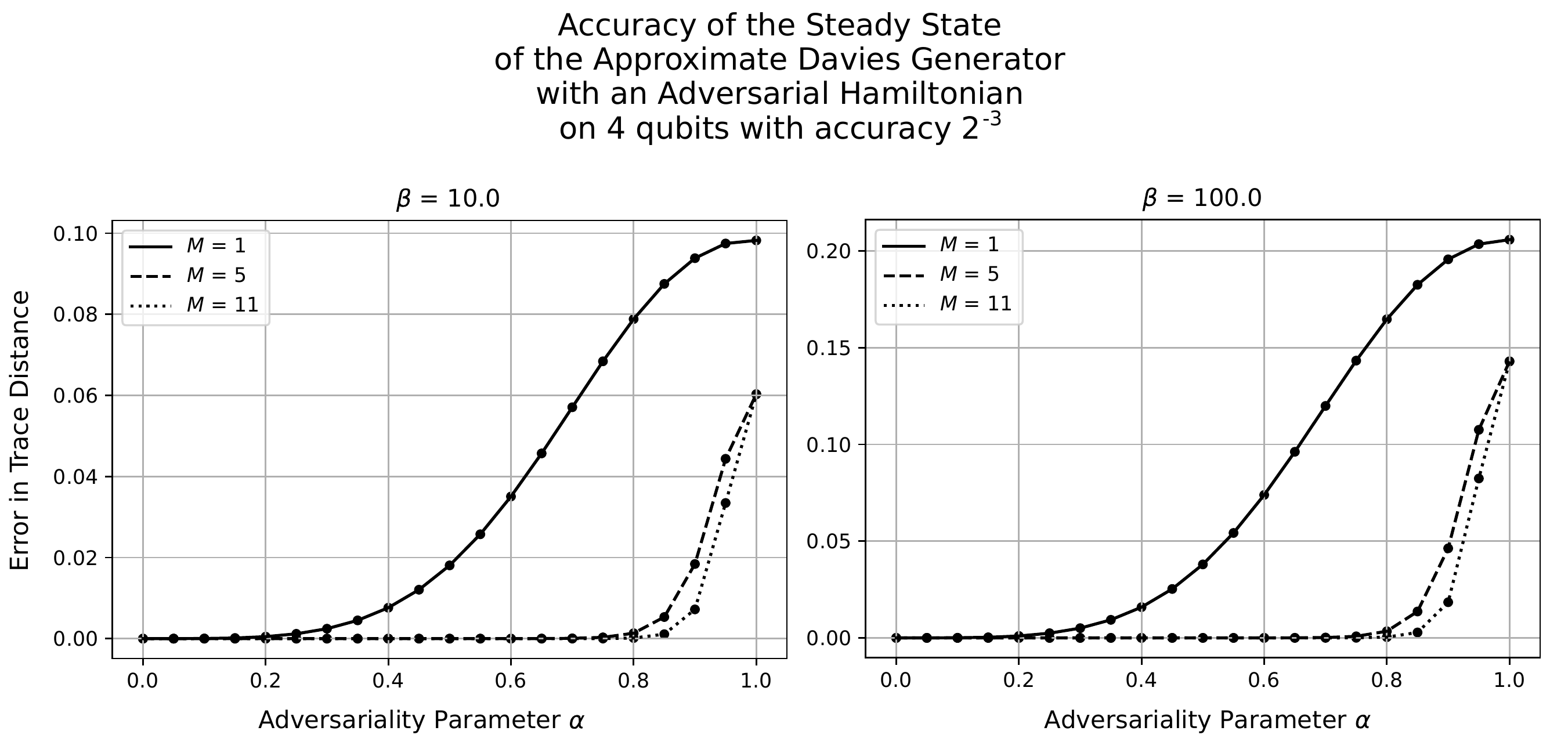}
    \caption{\label{fig:adversarial} Accuracy of thermal state preparation using the approximate Davies generator in \cref{eqn:apxDavies} with an adversarial Hamiltonian with eigenvalues as in \cref{eqn:adversarial}. We perform energy estimation to 3 bits of precision on a 4 qubit system. Energy estimation is performed with median amplification, where $M$ denotes the number of estimates over which the median is performed. We observe significant errors for large $\alpha$. When $M$ is increased, the large errors appear only for the largest $\alpha$. The software that produced these data is available at: \url{https://github.com/qiskit-community/promised-davies-generator}.}
\end{figure}

\bibliographystyle{alpha}
\bibliography{main}

\end{document}